\documentclass[acmtog, nonacm]{acmart}
\AtBeginDocument{%
  }

\setcopyright{none}
\acmDOI{XXXXXXX.XXXXXXX} 

\usepackage{algorithm}
\usepackage{algpseudocode}
\usepackage{pifont}
\usepackage{caption}
\usepackage{subcaption}
\newcommand{\cmark}{\ding{51}}
\newcommand{\xmark}{\ding{55}}
\usepackage{xcolor}

\begin{document}

\title{Deep Accurate Solver for the Geodesic Problem}

\author{Saar Huberman}
\email{saarhuberman@cs.technion.ac.il}
\affiliation{%
  \institution{Technion - Israel Institute of Technology}
  \city{Haifa}
  \country{Israel}
}

\author{Amit Bracha}
\email{amit.bracha@cs.technion.ac.il}
\affiliation{%
  \institution{Technion - Israel Institute of Technology}
  \city{Haifa}
  \country{Israel}}

\author{Ron Kimmel}
\email{ron@cs.technion.ac.il}
\affiliation{%
  \institution{Technion - Israel Institute of Technology}
  \city{Haifa}
  \country{Israel}}

\renewcommand{\shortauthors}{Huberman et al.}

\begin{abstract}
A common approach to compute distances on continuous surfaces is by considering a discretized polygonal mesh approximating the surface and estimating distances on the polygon. 
We show that exact geodesic distances restricted to the polygon are at most second-order accurate with respect to the distances on the corresponding continuous surface. 
By {\it order of accuracy} we refer to the convergence rate as a function of the average distance between sampled points. 
Next, a higher-order accurate deep learning method for computing geodesic distances on surfaces is introduced.
Traditionally, one considers two main components when computing distances on surfaces: a numerical solver that locally approximates the distance function, and an efficient causal ordering scheme by which surface points are updated. 
Classical minimal path methods often exploit a dynamic programming principle with quasi-linear computational complexity in the number of sampled points.
The quality of the distance approximation is determined by the local solver that is revisited in this paper. 
To improve state of the art accuracy, we consider a neural network-based local solver which implicitly approximates the structure of the continuous surface.
We supply numerical evidence that the proposed learned update scheme provides better accuracy compared to the best possible polyhedral approximations and previous learning-based methods. 
The result is a third-order accurate solver with a bootstrapping-recipe for further improvement.
\end{abstract}





  \begin{teaserfigure}
    \begin{center}
    \includegraphics[width=1\textwidth]{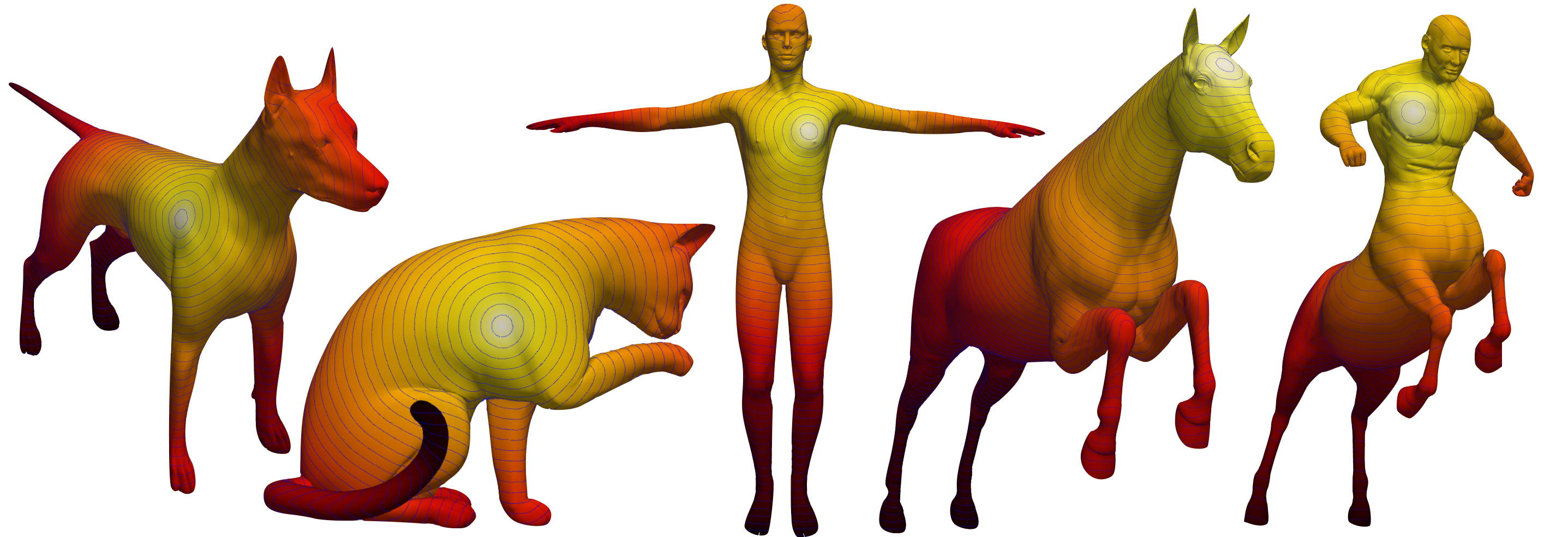}
    \caption{Geodesic distances from a single  point on each surface are presented as colors (white to red) and level sets painted on the surface. 
    The proposed method produces substantially more accurate than state of the art distance maps while operating in quasi-linear time.
    The evaluation was done on shapes from the TOSCA dataset \cite{bronstein2008numerical}, whereas our solver was trained with a small number of $2^{nd}$ order polynomial surfaces.
    }
    \end{center}
    \label{fig:contours}
    \Description{Figure 1. Geodesic distance computed by our method is shown on shapes from the Tosca dataset as distance maps and iso-contours. Our method computes highly accurate distances.}
  \end{teaserfigure}
\maketitle

\section{Introduction}

A geodesic distance is defined as the length of the shortest path connecting two points on a surface. 
It can be considered as a generalization of the Euclidean distance to curved manifolds.
The approximation of geodesic distances is used as a building block in many applications.
It can be found in robot navigation \cite{kimmel1998multivalued, kimmel2001optimal}, and shape matching \cite{ion20083d, elad2001bending, shamai2017geodesic, panozzo2013weighted}, to name just a few examples.
Thus, for effective and reliable use, computation of geodesics is expected to be both {\em fast} and {\em accurate}.

Over the years, many methods have been proposed for computing distances on polygonal meshes that compromise between the accuracy of the approximated distance and the complexity of the algorithm.
One family of algorithms for computing distances on polyhedral surfaces is based on solutions to the exact discrete geodesic problem introduced by Mitchell \textit{et al.} \cite{mitchell1987discrete}.  
The problem is defined as finding exact distances on a polyhedron. 
Mitchell\textit{ et al.} proposed an $\mathcal O(N^2 \log(N))$ complexity algorithm, where $N$ is the number of vertices of the polyhedral surface.
That algorithm, known as MMP, is computationally demanding and challenging to code, with the first implementation by Surazhsky \textit{et al.} \cite{surazhsky2005fast} introduced eighteen years after its publication.
At the other end, a popular family of methods for efficient approximation of distances known as {\it fast marching}, involves quasi-linear computational complexity $\mathcal O(N\log(N))$.
These methods are based on the {\it viscosity} solution of an {\it eikonal equation}. 
The fast marching method involve two main components, a heap sorting strategy and a local numerical solver, often referred to as a numerical update procedure.
The fast marching, originally introduced for regularly sampled grids \cite{sethian1996fast, tsitsiklis1995efficient}, was extended to triangulated surfaces in \cite{kimmel1998computing}.
While operating on curved surfaces approximated by triangulated meshes, the nearest neighbors of a vertex in the mesh are used to locally approximate a function whose gradient magnitude is equal to one.
This formulation of the unknown function is known as an eikonal equation. 
The resulting numerical algorithm yields a {\it first-order-accurate scheme} in terms of an average triangle's edge length $h$, denoted by $\mathcal O(h)$.

We start our journey by proving that the exact geodesic distances computed on a polygonal mesh approximating a continuous surface would be at most a second order approximation of the corresponding distances on the surface.
To overcome the second order accuracy limitation, we extend the numerical support about each vertex beyond the classical one ring approximation, and utilize the universal approximation properties of neural networks. 
The low complexity of the well-known dynamic programming update scheme \cite{dijkstra1959note}, combined with our novel neural network-based solver, yields an {\it efficient} and {\it accurate} method. 

In a related previous effort a neural network based local solver for the computation of geodesic distances was proposed \cite{lichtenstein2019deep}. 
We improve Lichtenstein's $\mathcal O(h^2)$ approach by extending the local neighborhood numerical support, and refining the network's architecture to obtain $\mathcal O(h^3)$ accuracy at similar linear complexity $\mathcal O(N\log(N))$.
In fact, directly extending the local support in \cite{lichtenstein2019deep} to 3rd ring neighborhood, does not improve the accuracy of the method.
It appears as if 90\% of that model's latent space is inactive. 
Based on these observations we propose to add hidden layers, change the activation functions, and reduce the size of the latent space. 

Similarly to \cite{lichtenstein2019deep}, the suggested solver is trained in a supervised manner using ground truth examples. 
And yet, since geodesics can not be derived analytically except for a limited set of surfaces like spheres and planes, we propose a multi-hierarchy bootstrapping technique. 
Namely, we use distance approximations on high resolution sampled meshes to better approximate distances at low resolutions.
We utilize our ability to apply given solvers at high resolution to generate higher order accurate training examples at low resolutions. 

\subsection{Contributions} 
We show that exact geodesics on polyhedrons are second order approximations. 
As a remedy, we develop a {\it fast} and {\it accurate} geodesic distance approximation method on surfaces.
\begin{itemize}
  \item 
  For fast computation, we use a distance update scheme (Algorithm \ref{alg:1}) that guarantees quasi-linear computational complexity. 
 \item 
 For accuracy, by revisiting the ingredients of the solver suggested in \cite{lichtenstein2019deep}, we propose a network based solver that operates directly on the sampled mesh vertices.
 \item 
 To provide accurate ground truth distances required for training our solver, we propose a novel data generation bootstrapping procedure. 
 \end{itemize}

\section{Related efforts}
Given a domain $\Omega \subset \mathbb{R}^n$ and a curve $\Gamma \in \Omega$, the predominant approach for generating distance functions from the curve $\Gamma$ to all other points in $\Omega$, is to find a function $\phi: \Omega \rightarrow \mathbb{R}$ which satisfies the {\it eikonal equation},
\begin{eqnarray}
\label{eq: eikonal equation}
| \nabla \phi(x) | &=& 1, \,\,\,\,\,\,\, x \in \Omega \setminus \Gamma \cr 
\phi(x) &=& 0, \,\,\,\,\,\,\, x \in \Gamma \,.
\end{eqnarray}
Due to the non-linearity and hyperbolicity of this partial differential equation (PDE), solving it directly is a challenge. 
Common solvers sample the continuous domain and approximate the solution on the corresponding discretized domain while being consistent with {\it viscosity} solutions \cite{crandall1983viscosity, rouy1992viscosity}.

\subsection{Fast Eikonal solvers}
In \cite{sethian1996fast, tsitsiklis1995efficient}, quasi-linear algorithms for approximating distances on regularly sampled grids were introduced.
These algorithms involve $\mathcal O(N\log(N))$ complexity, where $N$ is the number of points on the grid.
For example, the {\it fast marching algorithm} consists of two main parts, a numerical solver that locally estimates the distance function, by approximating a solution of an eikonal equation, and an ordering scheme that determines which points are visited at each iteration. 
Over the years, more sophisticated local solvers have been developed which utilize wide local support and lead to second \cite{sethian1999level} and third \cite{ahmed2011third} order accurate methods.
In \cite{kimmel1998computing}, Kimmel \& Sethian extended the fast-marching scheme to approximate first-order 
accurate geodesic distances on triangulated surfaces.
Another prominent class of numerical solvers is the fast sweeping methods \cite{kimmel2006method, zhao2005fast, li2008second, weber2008parallel}, iterative schemes that use alternating sweeping ordering.
These methods have a worst case complexity of $\mathcal O(N^2)$ \cite{hysing2005eikonal}.


\subsection{Geodesics in heat}
Instead of directly solving the hyperbolic eikonal equation, 
the heat method presented in \cite{crane2013geodesics} solves a linear elliptic PDE for heat transport. 
This method proposed a Poisson solver that aligns the unit gradients of a distance function with the gradient directions of a short time heat kernel, ending up with a first order accurate solver.
The distance functions produced with the heat method are smoother than those generated with fast marching and are less accurate near the discontinuities of the distance functions. 


\subsection{Window propagation methods} 
Trying to solve the discrete geodesic problem,  Mitchell et al. \cite{mitchell1987discrete} proposed a $\mathcal O(N^2 \log(N))$ complexity algorithm known as MMP. 
This algorithm was the first method introduced for computing exact distances on non-convex triangulated surfaces.
In its original formulation, it calculates the distance from a single source to all other points on the polygonal mesh.
The main idea of this algorithm is to track groups of shortest paths that can be atomically parameterized.
This is achieved by dividing each mesh edge into a set of intervals, which are referred to as \textit{windows}.
This quadratic algorithm is computationally demanding and challenging to code. 
In fact, the first implementation was introduced 18 years after its publication, by Surazhsky et al. \cite{surazhsky2005fast}. 
To reduce the computational complexity at the expense of accuracy, the framework presented in \cite{surazhsky2005fast} also supports approximated geodesics by merging windows during the update step.
Over the years, many improvements to the exact geodesic scheme were suggested \cite{xu2015fast, ying2014parallel, trettner2021geodesic}.
For example, the Vertex-oriented Triangle Propagation (VTP) algorithm \cite{qin2016fast}, which by sorting out superfluous windows is considered to be one of the fastest exact geodesics algorithms on polyhedral surfaces.

\subsection{Deep learning based methods}
A number of recent papers exploited neural networks approximation capabilities to numerically solve PDEs  \cite{greenfeld2019learning, hsieh2019learning, bin2021pinneik}.
Similar to our strategy, Lichtenstein et al. \cite{lichtenstein2019deep} proposed a deep-learning based method for geodesic distance approximation which we refer to as LPK.
It uses a heap sort ordering scheme while introducing a neural network based local solver.
For each distance evaluation of a target point $p$, a local neighborhood is obtained as input to the solver.
This neighborhood consists of all vertices connected to $p$ by a path with at most $2$ edges, which we refer to as second-ring neighborhood.
The proposed method showed second-order accuracy, similar to the exact geodesic method \cite{mitchell1987discrete}, while operating at quasi-linear $\mathcal O(N\log(N))$ computational complexity. 

\subsection{Curve shortening methods}
A geodesic is 
defined as the path connecting two points on a surface, which is characterized by having zero
geodesic curvature at each point along the curve.
Between two points on a surface, there can be multiple geodesic paths with different corresponding distances.
The geodesic distance is defined as the length of the minimal geodesic connecting each point on the surface to some source points at which the distance is defined to be zero.
A geodesic path can be extracted by starting with an arbitrary path on the surface and applying a length shortening method. 
Kimmel \& Sapiro \cite{kimmel1995shortening} presented a curve shortening flow that fixes the two endpoints of the curve at each iteration and minimize the geodesic curvature.
Sharp \& Crane \cite{sharp2020you} presented a curve shortening method for triangulated meshes by considering paths constrained to the mesh edges and applying intrinsic edge flips.
These local refinement methods converge an initial guess into a geodesic which is not necessarily the minimal one.
\begin{figure*}[htbp]
\begin{center}
\includegraphics[width=0.9\linewidth]{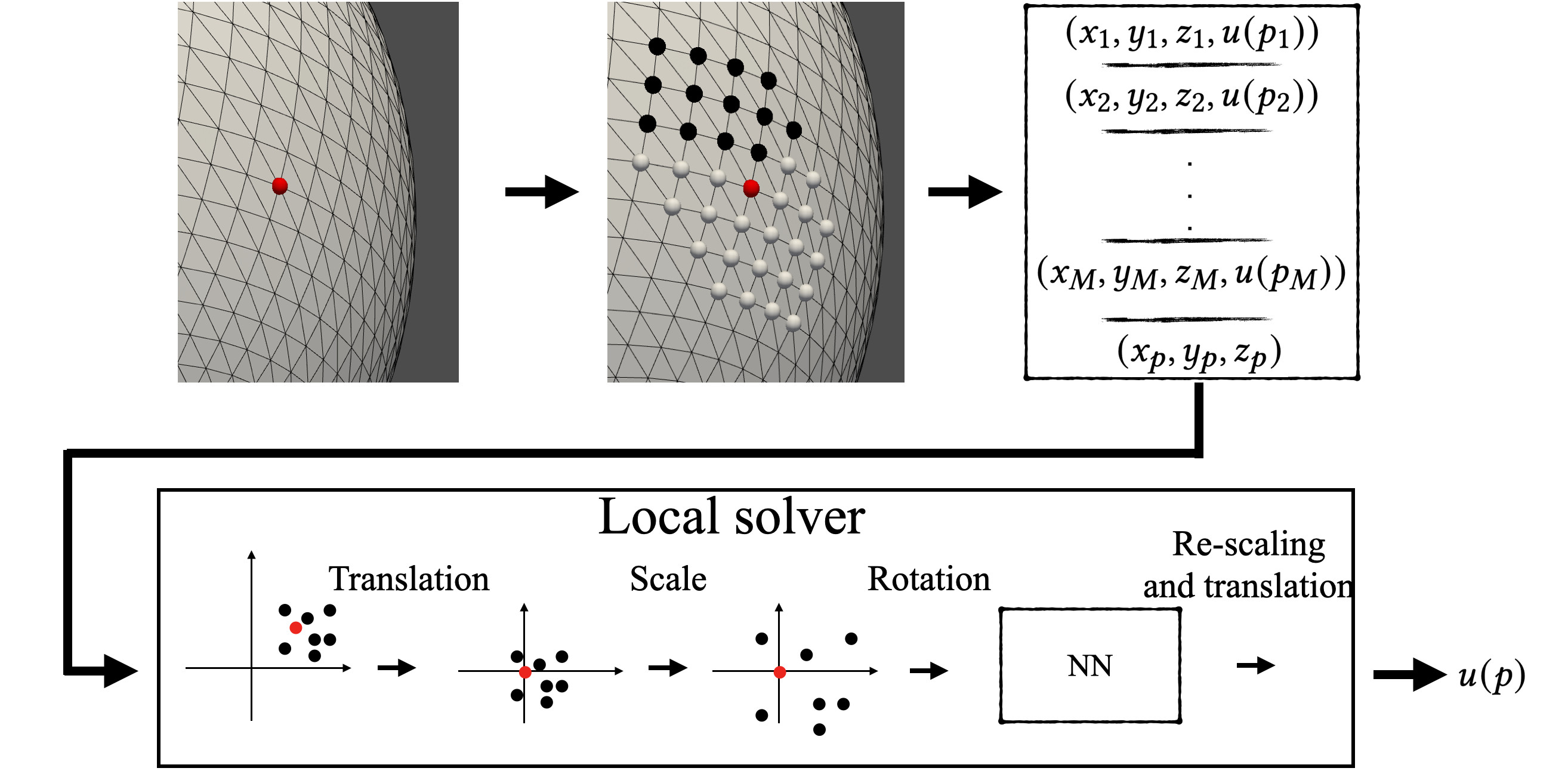}
\end{center}
  \caption{
  The proposed local solver pipeline described in Section \ref{local_Solver_section}.
  For each point in the {\it Wavefront}, the {\it Visited } points from from its local neighbourhood are chosen, and their coordinates and distance function are transformed into their canonical representation with respect to the target point. 
  The processed input is fed into the neural network and the output is further processed to revert the centering and scaling transformations.
  }
\label{fig:solver pipeline}
\Description{Figure 2. illustration of the pipeline of our local solver. Fully described in the text.}
\end{figure*}
\section{Exact distances on polyhedral surfaces}
\begin{theorem}
\label{thm:exact_is_2nd}
Let ${\mathcal M}:\Omega\in \mathbb{R}^2 \rightarrow \mathbb{R}^3$ be a Riemannian two dimensional manifold with effective Gaussian curvature a.e.
Let $C(s):[0,L]\rightarrow {\mathcal M}$ be a minimal geodesic connecting two surface points $C(0)$ and $C(L)$ on $\mathcal M$ with arclength parametrization $s$, and $L$ the length of $C$.
The length of $C$  differs by $\mathcal{O}(h^2)$ from the sum of the lengths of the cords defined by sampling the curve at equal distances (in the Euclidean embedding spaces) between the points. 
These line segments,  of length $h$ each as measured in $\mathbb{R}^3$, are defined by a sequence of surface points $C(s_{i})$ and $C(s_{i+1})$. 
That is, the length of the approximation $\gamma$ defined by its vertices
$ \{C(0),C(s_2),\ldots,C(s_{n-1}),C(L)\}$, given by
\begin{eqnarray}
L(\gamma) &=&  \sum_{i=1}^{n-1} \|C(s_{i+1})-C(s_i)\|_{\mathbb{R}^3}  
\,=\,nh, \,
\end{eqnarray}
 differs by $\mathcal{O}(h^2)$ from 
 \begin{eqnarray}
L(C) &=& \int_0^L  ds.
\end{eqnarray}
\end{theorem}
\begin{proof}[Proof of Theorem \ref{thm:exact_is_2nd}]
 Consider the length parameterization along the line segment with end points $C(s_i)$ and $C(s_{i+1})$ be given by $t \in [-h/2,h/2]$, and assume w.l.o.g. the monotone increasing reparametrization $s(t)$ that would allow us to parametrize the surface geodesic segment between $C(s_i)$ and $C(s_{i+1})$.
 As $t$ is the arclength along the cord connecting the two end points of the line segment, by freedom of parametrization, we could choose $|C_t(0)| = 1$. 
 
 Next, let us compute the length of $C(t)$ in the $i$th interval,
 \begin{eqnarray}
  L_{\mathcal M}(C(s_i),C(s_{i+1})) & =&
  \int_{s_i}^{s_{i+1}} ds 
  \, =\, \int_{-h/2}^{h/2} |C_t| dt. \,
\end{eqnarray}
 Let us expand $|C_t|$ about $0$, by which we have 
{\small
\begin{eqnarray}
 |C_t(t)| &=& |C_t(0)|+t \left ( \frac{d}{dt}|C_t| \right )(0) + \frac{t^2}{2} \left ( \frac{d^2}{dt^2}|C_t|\right )(0) + \cdots
 \cr &=& 
 |C_t(0)|+t \left ( \frac{\langle C_t, C_{tt}\rangle }{|C_t|} \right )(0) + \frac{t^2}{2} \left ( \frac{d}{dt}\frac{\langle C_t, C_{tt}\rangle }{|C_t|}\right )(0) \cr 
 && +
 \cdots
 \end{eqnarray}
} 
 Let us focus on the second  term,
\begin{eqnarray}
  \frac{\langle C_t, C_{tt}\rangle}{|C_t|}& =&
    \frac{\langle C_t, C_{tt}\rangle}{|C_t|^3}|C_t|^2
\,\,    = \,\, \kappa |C_t|^2 ,
\end{eqnarray}
 where $\kappa$ 
  is the curvature (normal curvature for a geodesic) of $C$ at that point.  
 The third term is given by
 \small{
 \begin{eqnarray}
  \frac{d}{dt}  \kappa |C_t|^2  &=& 
  \kappa_t |C_t|^2 + 2\kappa\langle C_t,C_{tt}\rangle 
  \,\,=\,\, \kappa_s |C_t|^3 + 2\kappa^2 |C_t|^3.
\end{eqnarray}
}
 
 We conclude with  
 \begin{eqnarray}
  L_i &=&\int_{s_i}^{s_{i+1}} ds 
  \, =\, \int_{-h/2}^{h/2} |C_t| dt \cr
  &=&
\int_{-h/2}^{h/2} \Big( |C_t(0)|+t \left ( \kappa |C_t|^2 \right )(0) \cr
  && +
\frac{t^2}{2}\left ( \kappa_s |C_t|^3 + 2\kappa^2 |C_t|^3 \right )(0) + \cdots \Big) dt \cr
&=& |C_t(0)|h +  \left ( \kappa_s |C_t|^3 + 2\kappa^2 |C_t|^3\right )(0)\frac{h^3}{24}+ \mathcal{O}(h^5).
 \end{eqnarray}
 With our specific selection of $|C_t(0)| =  1$,
 we conclude with the overall error given by
 {\small 
\begin{eqnarray}
\mathcal{E}rr \,= \, \sum_{i=1}^{n-1} |L_i - h| 
  \, =\,  \sum_{i=1}^{n-1} \left |   (\kappa_s+2\kappa^2) \frac{h^3}{24} + \mathcal{O}(h^5)\right |
  \,=\, \mathcal{O}(h^3) \mathcal{O}(n), \, 
\end{eqnarray}
}
 where $\kappa$ and $\kappa_s$ are evaluated at $t=0$ for each segment. 
Note, that  $\kappa$ and $\kappa_s$ are geometric quantities and thus could be regarded as effective bounded constants.
Then, assuming $h \approx \mathcal{O}(n^{-1})$ we have the convergence rate to be $\mathcal{O}(h^2)$.
\end{proof}


\subsection{Example: Circle  perimeter length estimator} 
Assume we would like to approximate the length of a circumference of a circle with radius $1$ in the plane using a regular polygon with $N$ vertices. 
Let $\theta = \frac{2\pi}{N}$ be the angle of the circular sector defined between two successive sample points on the circle.
The distance between these points is given by 
$h = 2 \sin\left (\frac {\theta}{2}\right )$.
The circumference of the circle is known to be $2\pi$, while the length approximated by the polygon is $N h = 2N \sin\left (\frac{\pi}{N}\right )$.
The truncation error is then given by
\begin{eqnarray}
 2\pi - N h &=& 2\pi - 2N \sin\left (\frac{\pi}{N}\right ) \cr
&=& 2\pi - 2N \left(\frac{\pi}{N} 
- \frac{\left( \frac{\pi}{N} \right)^3}{3!} 
+ \frac{\left( \frac{\pi}{N} \right)^5}{5!}  - \cdots \right) \cr
&=& \frac{\pi^3}{3 N^2} - \frac{\pi^5}{60 N^4} + \cdots \cr
&=& \mathcal{O}\left (N^{-2}\right ) \cr 
 & =& \mathcal{O}(h^2).
\end{eqnarray}

Note, that this analysis also provides a lower bound on the length estimation error of great circles on a sphere.
\section{Geodesics: ${\mathcal O}(h^3)$ accuracy  at ${\mathcal O}(N \log N)$ complexity}
We present a neural network based method for approximating accurate geodesic distances on surfaces. 
Similar to most dynamic programming methods, like the fast marching scheme, the proposed method consists of a numerical solver that locally approximates the distance function $u$ at a surface point $p$, and an ordering scheme that defines the order of the visited points. 
Here, the $N$ sampled surface points are divided into three disjoint sets. 
\begin{enumerate}
    \item \textit{Visited}: points where the distance function $u(p)$ has already been computed and will not be changed.
    \item \textit{Wavefront}: points where the computation of $u(p)$ is in progress and is not yet fixed.
    \item \textit{Unvisited}: points where  $u(p)$ has not yet been computed.
\end{enumerate}
\begin{algorithm}
\caption{Distance Updating Scheme}\label{alg:1}
\begin{algorithmic}[1]
\algnewcommand{\Initialize}[1]{%
  \State \textbf{Initialize:}
  \Statex \hspace*{\algorithmicindent}\parbox[t]{.8\linewidth}{\raggedright #1}
}
\algnewcommand{\Define}[1]{%
  \State \textbf{Definitions:}
  \Statex \hspace*{\algorithmicindent}\parbox[t]{.8\linewidth}{\raggedright #1}
}
\Define{$S$ - Set of all source points \\
$p$ - point on the surface\\
$u(p)$ - minimal distance from sources to $p$}
\Initialize{$u(p) = 0$, tag $p$ as {\it Visited};
\ $\forall p \in S$ \\
$u(p) = \infty$, tag $p$ as {\it Unvisited}; 
    $\forall p \not \in S$ \\
    Tag all {\it Unvisited} points adjacent to {\it Visited} points as {\it Wavefront}
    }
\Repeat
    \For{$p \in  {\mbox{\it Wavefront}}$}
        \State Approximate $u(p)$ based on {\it Visited} points
    \EndFor
    \State \parbox[t]{0.9\linewidth}{ Tag the least distant {\it Wavefront} point $p'$ as {\it Visited}}
    \State Tag all {\it Unvisited} neighbors of $p'$ as {\it Wavefront}
\Until{all points are {\it Visited}.}
\State \textbf{Return} $u$
\end{algorithmic}
\label{alg:1}
\end{algorithm}
\begin{figure*}[htbp]
\begin{center}
\includegraphics[width=0.9\linewidth]{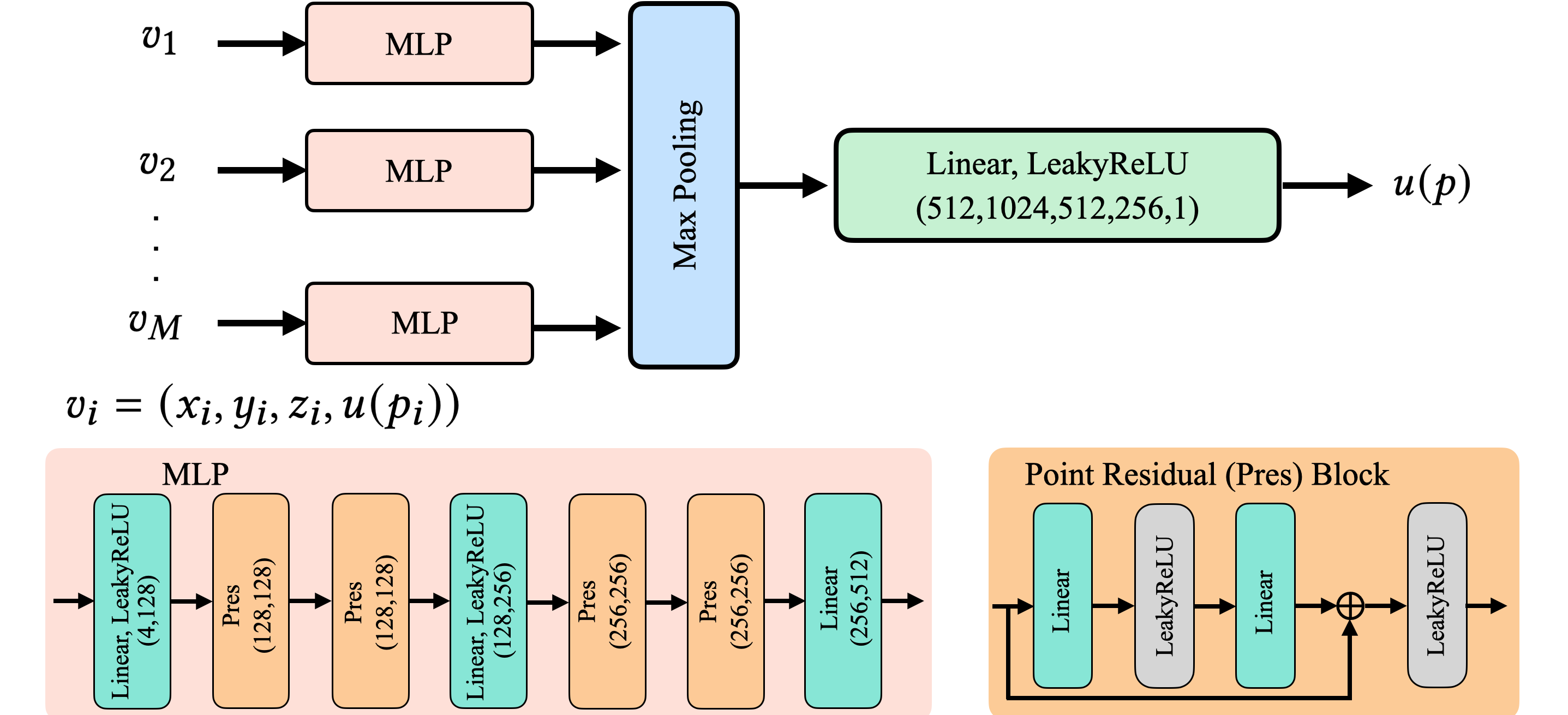}
\end{center}
  \caption{The proposed network architecture described in Section \ref{local_Solver_section}.
  The input coordinates and distance function are in their canonical form, after translation rotation and scale.
  }
\label{fig:NN schematic}
\Description{Figure 3. Illustration of the architecture of our local solver. Fully described in the text.}
\end{figure*}
The distances at the sampled surface points are computed according to Algorithm \ref{alg:1}, where Step 5 of the scheme is performed by the proposed local solver. 
When applied to a target point $p\in$ {\it Wavefront}, the local solver uses a predefined maximum number of {\it Visited} points. 
These {\it Visited} points are chosen from the local neighborhood and are not related to the number of points on the mesh; hence, a single operation of our solver has constant complexity.
Since, within our dynamic programming setting, the proposed method retains the heap sort ordering scheme, the overall computational complexity is $\mathcal O(N\log(N))$.
Subsection \ref{local_Solver_section} introduces the operation of the local solver, presents the required pre-processing and elaborates on the implementation of the neural network. 
Next, subsection \ref{train_local_solver_section} explains how the dataset is generated and the network weights are optimized. 
Finally, subsection \ref{generate_gt_section} details how ground truth distances are calculated when no analytic closed form solution is available.

\subsection{Local solver} 
\label{local_Solver_section}
Here, we introduce a local neural network-based solver
for Step 5 of Algorithm \ref{alg:1}.
When the solver is applied to a given point $p\in $ {\it Wavefront}, it receives as input the coordinates and distance function values of its neighboring points.
The neighboring points, denoted by $\mathcal N(p) = \{p_1,p_2,...,p_M\}$, are defined by all vertices connected to $p$ by a path of at most $3$ edges, which is often referred to as third ring neighborhood.
Based on the information from the {\it Visited} points in $\mathcal N(p)$, the local solver approximates the distance function $u(p)$.
This way, we keep utilizing the order of updates that characterizes the  construction of distance functions.
As mentioned earlier, for a given target point $p$ and neighboring points $\{p_i\}_{i=1}^M \subset$ {\it  Visited} $ \cap \mathcal N(p)$, the input to our solver is $\{(x_{p_i}, y_{p_i},z_{p_i}, u(p_i)\}_{i=1}^M \cup \{(x_{p}, y_{p},z_{p})\}$. 
To address the solver's generalization capability and to handle diverse possible inputs, we transform the input to the neural network into a canonical representation. 
To this end, we design a preprocessing pipeline presented in Figure \ref{fig:solver pipeline}.
The coordinates are centered with respect to the target point, resulting in relative coordinates $(x_{p_i}-x_p, y_{p_i}-y_p, z_{p_i}-z_p)$, and ${\min_j \{u(p_j)}\}$ is subtracted from the values of the distance function $\{u(p_i)\}_{i=1}^M$.
After the input is centered, it is scaled so that the mean L2 norm of the coordinates is of unit size. 
Last, a $SO(3)$ rotation matrix is applied to the coordinates so that their first moment is aligned with a predefined principal directions.
The processed input is fed into the neural network and the output is further processed to reverse the centering and scaling transformations. 

The input neighborhood has no fixed order and can be viewed as a set. 
To properly handle our unstructured set of points, we train our neural network output to be permutation invariant.
The network's architecture, presented in Figure \ref{fig:NN schematic}, consists of three main components. 
 A shared weight encoder that lifts the $4$-dimensional input to $512$ features using residual multi-layer perceptron (MLP) blocks \cite{ma2022rethinking}.
 A per-feature max pooling operation that results in a single $512$ feature vector, and 
 a fully connected regression network of dimensions $(512, 1024, 512, 256, 1)$ that outputs the desired target distance.
 \begin{figure*}[htbp]
\begin{center}
\includegraphics[width=0.9\textwidth]{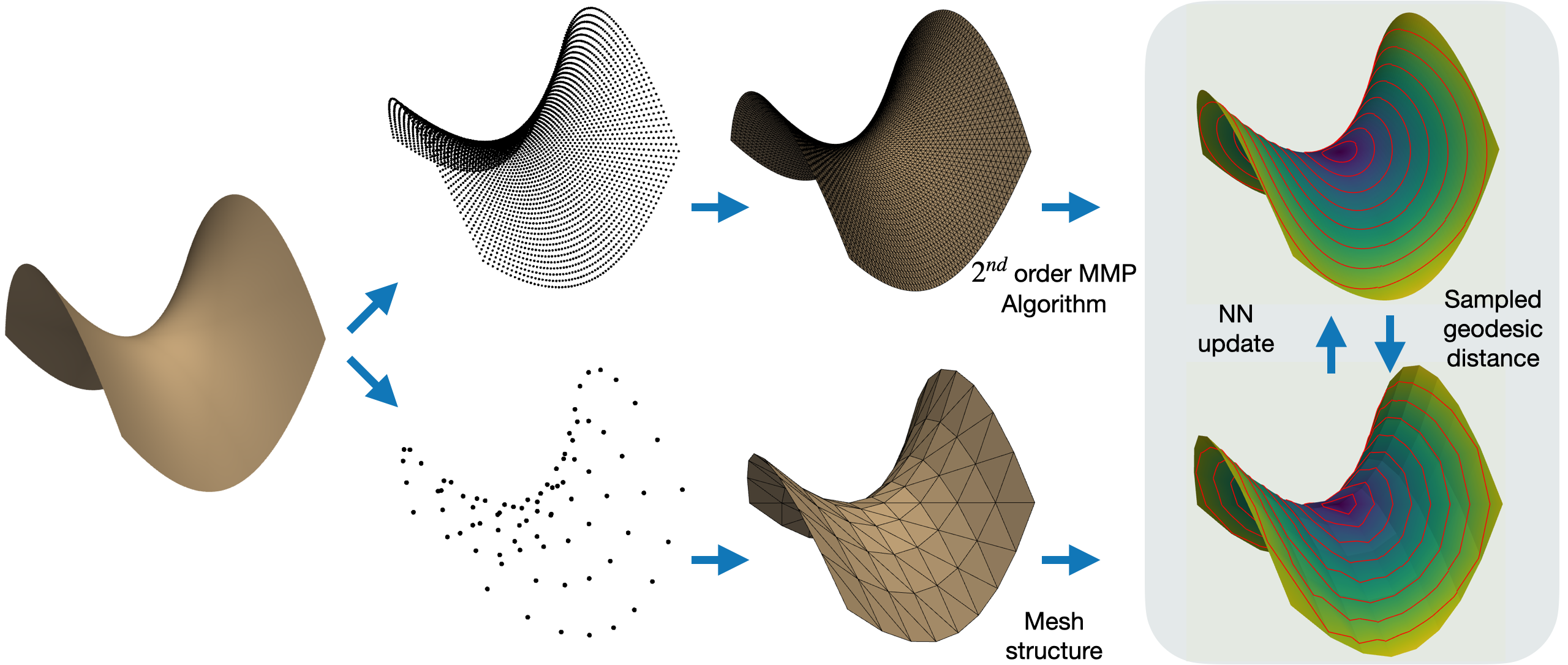}
\end{center}
  \caption{Bootstrapping by training. Distance values computed for a high $h^2$-resolution sampled mesh of a continuous surface with an $r$ accurate scheme  yields $\mathcal{O} (h^{2r})$ accurate distances given at the mesh points.
  The mesh can then be sampled into a lower $h$-resolution mesh of the same continuous surface, while keeping the corresponding $\mathcal{O} (h^{2r})$ accurate distances at the vertices.
  See text for an elaborated discussion regarding data augmentation at high resolution and training more accurate update procedures at the low resolution.
  }
\label{fig:gt_generation_scheme}
\Description{Figure 4. Present the proposed bootstrapping method used to create our training examples. Fully described in the text.}
\end{figure*}
\subsection{Training the local solver} 
\label{train_local_solver_section}
We use a customary supervised training procedure, using examples with corresponding ground-truth distances.
These ground truth distances are obtained by applying a bi-level sampling strategy, as detailed in
Section \ref{generate_gt_section}.
Given an input $\{(x_{p_i}, y_{p_i},z_{p_i}, u(p_i)\}_{i=1}^M$, our network is trained to minimize the difference between its output and its corresponding ground truth, denoted by $u_{gt}(p)$.
To develop a reliable and robust solver, we create a diverse dataset that simulates a variety of scenarios. 
We construct this dataset by selecting various source points and sampling local neighborhoods at different random positions relative to the sources.
According to the causal nature of our algorithm, we build our training examples, such that a neighboring point $p'$ is defined as {\it Visited} and is allowed to participate in the prediction of $u(p)$ if $u_{gt}(p') < u_{gt}(p)$. 
The network's parameters $\Theta$ are optimized to minimize the Mean Square Error (MSE) loss
\begin{eqnarray}
\label{eq:MSE loss}
L(\Theta) &=& \frac{1}{K}\sum_{j=1}^K{(f_{\Theta}(\{(x_{p_{i,j}}, y_{p_{i,j}},z_{p_{i,j}}, u(p_{i,j}))\}_{i=1}^M)-u_{gt}(p_j))^2},\cr
&&
\end{eqnarray}
where $K$ is the number of examples in the training set and $p_{i,j}$ corresponds to the $i^{th}$ neighbor of the target point $p_j$.
The coordinates and distances used in our training procedure are in their canonical form, after being translated, rotated and scaled, as explained in Section \ref{local_Solver_section}.
\begin{figure*}[htbp]
\begin{center}
\includegraphics[width=0.8\textwidth]{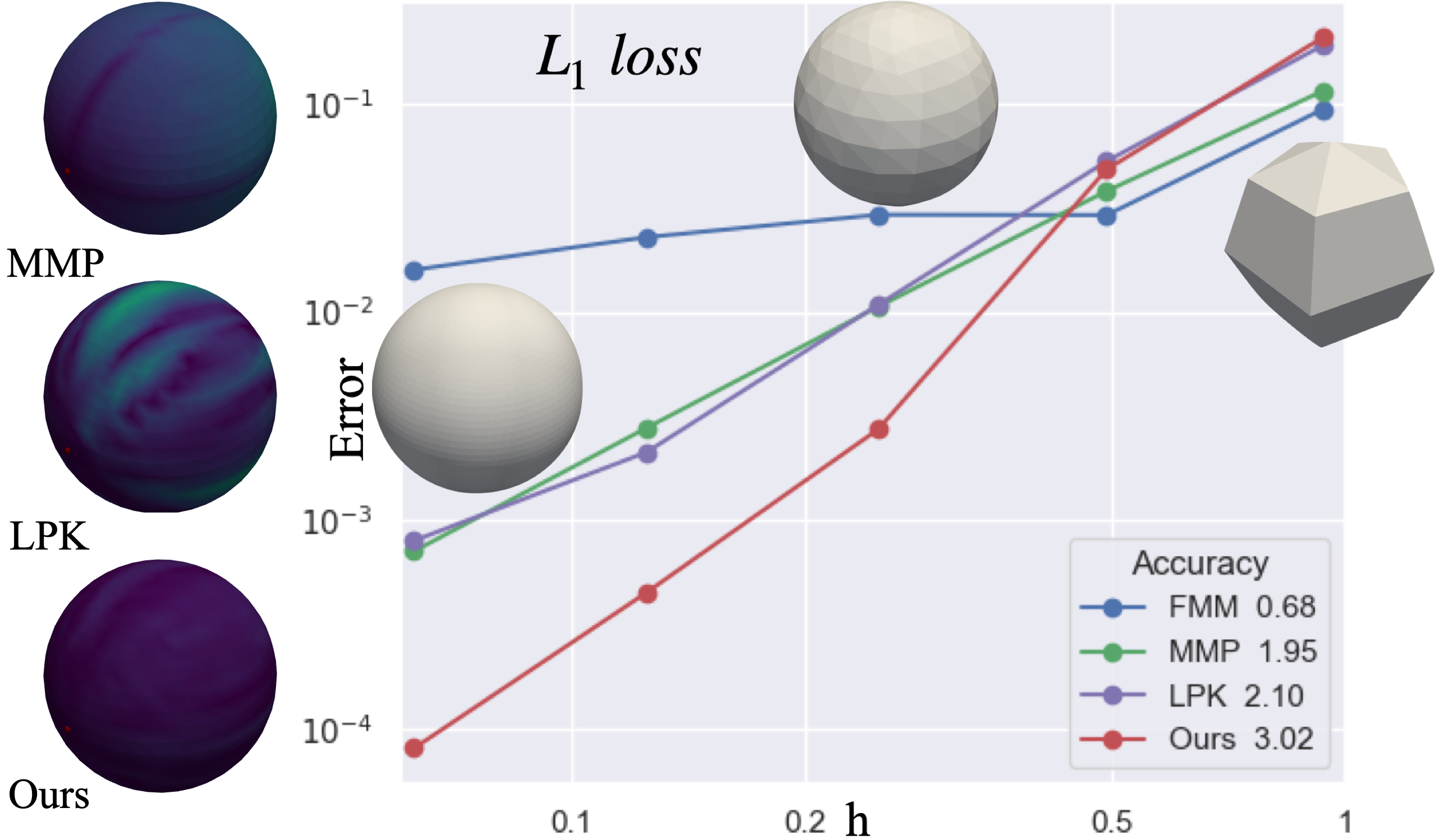}
\end{center}
  \caption{Order of accuracy: Right: Plots showing how the edge size effects the error.
     The accuracy of each scheme is defined by the corresponding slope. 
     We used mesh approximations of a unit sphere with different edge size. 
    Left: Errors presented on high resolution spheres for 2nd order methods and our presented method (when brighter colour indicates higher error).
  }
\Description{Figure 5. Numerical evaluation of the order of accuracy presented on spheres. The order of accuracy for the fast marching method is 0.68, for the MMP method 1.95, for the method presented by Lichtenstein et al. 2.1, and for our method 3.02. For high-resolution meshes with an average edge of less than 0.1, our method shows errors that are more than an order of magnitude smaller compared to all other methods}
\label{fig:spheres}
\end{figure*}
\subsection{Learning to augment}
\label{generate_gt_section}
Exact distances on continuous surfaces are given by analytic expressions for a very limited set of continuous surfaces; namely, for spheres and planes.
Since our solver is trained on examples containing ground truth distances, an additional approximation algorithm must be considered to generate our training examples for general surfaces.
Currently, the most accurate axiomatic method for distance computation is the MMP algorithm, which computes ``exact'' polyhedral distances.
Considering polyhedral surfaces as sampled continuous ones, the ``exact'' distances on triangulated surfaces are $2^{nd}$ order accurate with respect to the edge length $h$. 
Indeed, the MMP is an $\mathcal{O}(h^2)$ accurate method.
In order to train our network with more accurate than $\mathcal{O}(h^2)$ distances for general smooth surfaces, we resort to the following bootstrapping idea. 

We introduce a multi-resolution ground truth boosting generation technique that allows us to obtain ground truth distances of any desired order.
The underlying idea is that distances computed on a mesh obtained from a denser sampling of the surface are a better approximation to the distances on the continuous surface.
When generating examples from a given surface, two sampling resolutions of the surface are obtained and corresponding meshes are formed, denoted by ${\mathcal M}_{dense}$ and $ {\mathcal M}_{sparse}$, respectively. 
Distances are computed on the high-resolution mesh ${\mathcal M}_{dense}$ and the obtained distance map is sampled at ${\mathcal M}_{sparse}$.
Consider $h_{dense}, h_{sparse}$  that correspond to the mean edge length of the polygons ${\mathcal M}_{dense}, {\mathcal M}_{sparse}$, so that,
$ h_{dense}\,= \, h_{sparse}^q \,. $
The distances computed by an approximation method of order $r$ on ${\mathcal M}_{dense}$ are $r$ order accurate $\mathcal{O}(h_{dense}^r)$. 
Therefore, the same approximated distances, sampled at the corresponding vertices of ${\mathcal M}_{sparse}$, have $\mathcal{O}(h_{sparse}^{qr})$ accuracy.

Using the distance samples of the polyhedral distances obtained by the MMP algorithm while requiring $q\geq2$, allows us to generate distance maps that are at least fourth-order accurate.
By considering these approximated distances as our ground truth, training examples can be generated from ${\mathcal M}_{sparse}$ as described in Section \ref{train_local_solver_section} which allow us to properly train a third-order accurate method.
The iterative application of this process allows us to generate accurate ground truth distances to properly train solvers of arbitrary order.
For example, after training a $3^{rd}$ order solver, we can apply the same process while replacing the MMP with our new solver to generate a $\mathcal{O}(h^6)$ ground truth distances and train a solver up to $6^{th}$ order.
For a schematic representation of this technique, see Figure \ref{fig:gt_generation_scheme}.
\section{Numerical evaluation: Spheres and beyond}
Geodesic distances on spheres can be calculated analytically.
Therefore, they are well suited for the evaluation of our method.
For two given points 
$a = (x_a, y_a, z_a),\, b = (x_b, y_b, z_b)$ lying on a sphere of radius $r$, the geodesic distance between them is defined by 
\begin{eqnarray}
\label{eq:spheres geodesic}
u(a,b) &=& r \ \mbox{arccos}\left (\frac{x_a x_b + y_a y_b + z_a z_b}{r^2}\right ) \,.
\end{eqnarray}

To train our solver, we first randomly sampled spheres at different resolutions and obtained triangulated versions of them.
Using the exact distances we created a data set with 100,000 examples and applied a training procedure as presented in Section \ref{train_local_solver_section}.
To evaluate our method, we constructed a hierarchy of spheres with various resolutions, see Figure 
\ref{fig:spheres}.

As described in \cite{osher1988fronts}, we assume that the exact solution $u(a,b)$ can be written as
\begin{eqnarray}
\label{eq:solution approximation}
u(a,b) \,=\, u_h(a,b) + Ch^R + \mathcal{O}(h^{R+1}), \,
\end{eqnarray}
where $C$ is a constant, $u_h$ is defined as the approximate solution on a mesh with a corresponding mean edge length of $h$, and $R$ is the order of accuracy.
For two given mesh resolutions of the same continuous surface ${\mathcal M}_1, {\mathcal M}_2$ with corresponding $h_1, h_2$, we can estimate our method's order of accuracy by
\begin{eqnarray}
\label{eq:order of accuracy}
R &=& \log_{\frac{h_1}{h_2}}
\left(\frac{u-u_{h_1}}{u-u_{h_2}}\right ) \,.
\end{eqnarray}
The evaluation of our method is shown in Figure \ref{fig:spheres}, where the slope of each graph indicates the order of accuracy $R$.
It can be seen that our method has a higher order of accuracy than the classical fast marching (FMM), the MMP exact geodesic method, and the previous deep learning method proposed by Lichtenstein et al.
\begin{figure}[t]
\centering
\includegraphics[width=\linewidth]{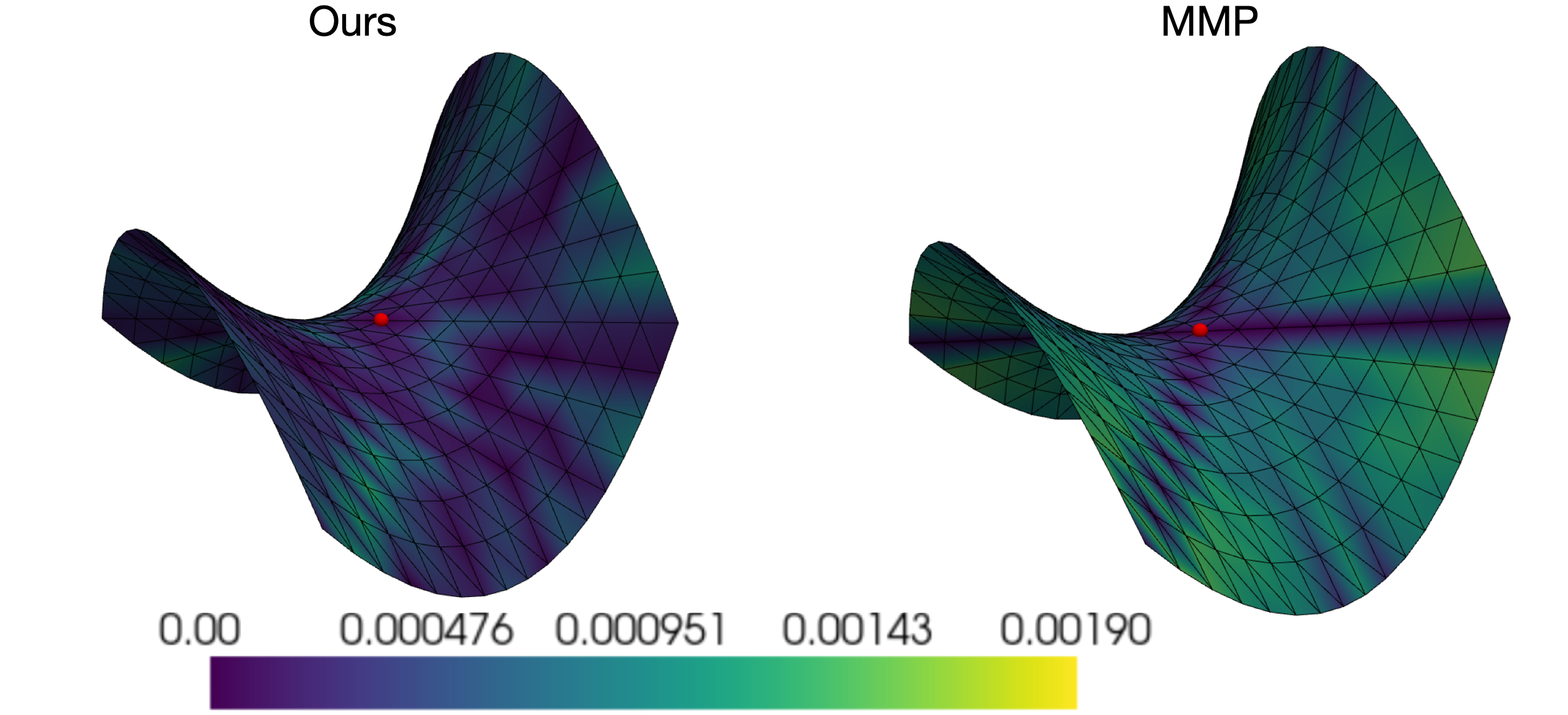} 
\caption{{Polynomial surfaces:} Errors presented for the polyhedral scheme and the proposed method. 
   Local errors, represented as colors on the surface, were computed relative to exact polyhedral distances computed at a high-resolution sampled mesh of the continuous surface, as described in Subsection \ref{generate_gt_section}.
   }
\label{fig: saddle_errors}
\Description{Figure 6. Visualized errors on hyperbolic parabolic surfaces comparing MMP methods and our method. Our method has significantly lower errors compared to MMP.}
\end{figure}
\begin{table*} 
    \centering
    \captionof{table}{{Polynomial surfaces:} Quantitative evaluation conducted on $2^{nd}$ order paraboloids. The errors were computed relative to the polyhedral distance projected from high-resolution sampled meshes, as described in Subsection \ref{generate_gt_section}.}
    \scalebox{1}{\begin{tabular}{|l|c|c|c|c|c|c|c|c|}
    \hline
    Surface & \multicolumn{4}{|c|}{$L_1$ \quad $ (\times10^{-2})$} & \multicolumn{4}{|c|}{$L_\infty$ \quad $ (\times10^{-2})$}\\\cline{2-9}
    &\quad  FMM  \quad  & \quad LPK \quad & \quad MMP \quad & \quad Ours \quad & \quad FMM \quad &  \quad LPK \quad & \quad MMP \quad & \quad Ours \quad \\
    \hline\hline\hline
    $x^2-y^2$ & 2.86 & 0.21 & 0.09 & \textbf{0.04} & 7.95 & 1.17 & 0.24 & \textbf{0.16}\\
    $x^2+y^2$ & 2.05 & 0.34 & 0.28 & \textbf{0.09} & 6.80 & 1.44 & 0.67 & \textbf{0.26}\\
    $x^2-y^2+xy$ & 2.91 & 0.33 & 0.18 & \textbf{0.09} & 6.40 & 2.60 & 0.63 & \textbf{0.28}\\
    \hline
    \end{tabular}} 
\label{table: polynomial surfaces}
\end{table*}
\begin{figure*}[htbp]
\begin{center}
\includegraphics[width=\textwidth]{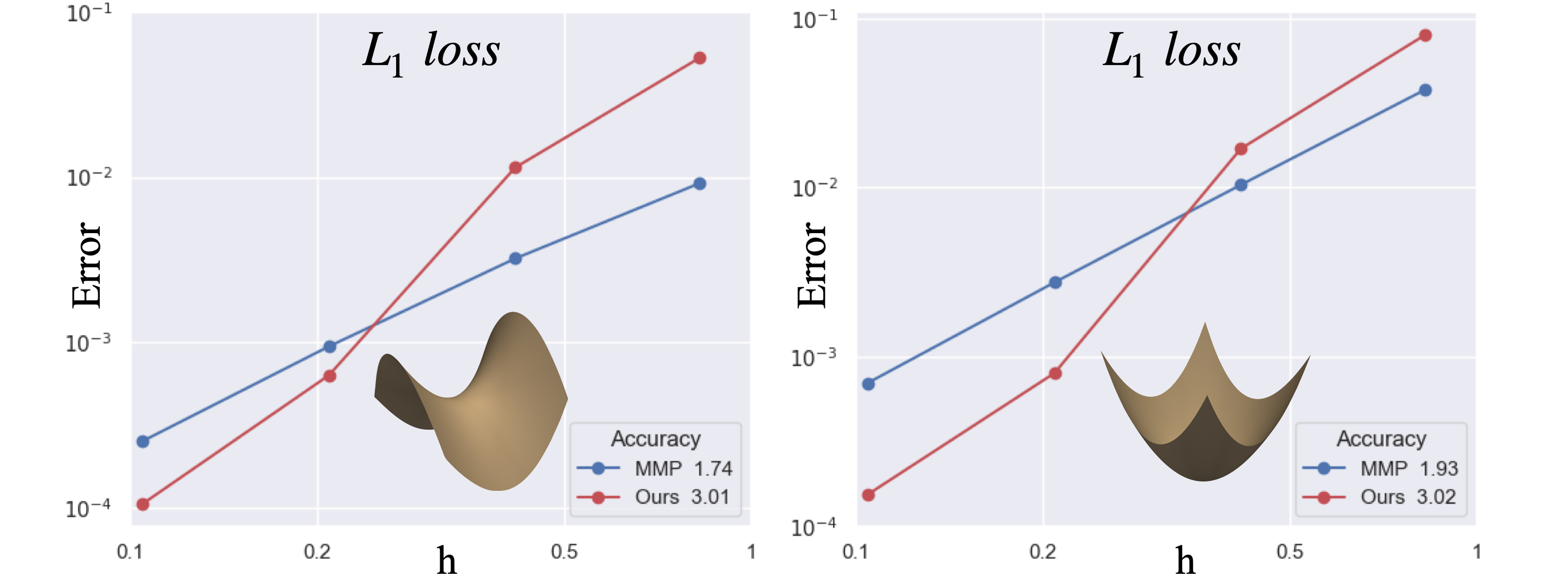}
\end{center}
\caption[Order of accuracy presented on parametric surfaces]
{{Order of accuracy on parametric surfaces:} Plots showing the effect of the edge size on the errors. 
The accuracy of each scheme is associated with its corresponding slope. 
Ground truth distances were evaluated using our bootstrapping method \ref{generate_gt_section}. 
Left: Evaluation on hyperbolic paraboloid $x^2-y^2$.
     Right: Evaluation on regular paraboloid $x^2+y^2$.
     }
\Description{Figure 7.  Numerical evaluation of the order of accuracy for a regular paraboloid and a hyperbolic paraboloid. For a regular paraboloid, the order of accuracy is 1.93 and 3.02 for the MMP method and our method, respectively. For a hyperbolic paraboloid, the accuracy is 1.74 for the MMP method and 3.01 for our method.}
\label{fig: 3rd_order_polynom}
\end{figure*}

\begin{table*}[htbp] 
    \centering
    \captionof{table}{Generalization to arbitrary surfaces: Quantitative evaluation tested on TOSCA. 
    The errors are relative to the MMP scheme. 
    Our solver and the one proposed in \cite{lichtenstein2019deep} were trained using our bootstrapping method \ref{generate_gt_section} with a limited number of $2^{nd}$ order polynomial surfaces, the three paraboloids  presented in Table \ref{table: polynomial surfaces}. 
    In addition to shapes represented as triangulated meshes, we present results for applying the proposed method to point clouds, see Section \ref{point_cloud_section}.}
    \scalebox{0.94}{\begin{tabular}{|l|c|c|c|c|c|c|c|c|c|c|}
    \hline
    Shape & \multicolumn{5}{|c|}{$L_1$ \quad $ (\times10^{-1})$} & \multicolumn{5}{|c|}{$L_\infty$ \quad $ (\times10^{-1})$}\\\cline{2-11}
    & \quad Heat \quad & \quad FMM \quad & \quad LPK \quad & \quad Ours \quad & \quad Ours - Point clouds \quad & \quad Heat \quad & \quad FMM \quad& \quad LPK \quad & \quad Ours \quad & \quad Ours - Point clouds \quad \\
    \hline\hline\hline
    Dog & 0.728 & 0.110 & 0.123 & \textbf{0.037} & 0.072 & 8.688 & 1.514 & 1.318 & \textbf{0.465} & 2.968\\
    Cat & 1.260 & 0.101 & 0.278 & \textbf{0.043} & 0.045 & 5.379 & 0.735 & 3.949 & \textbf{0.474} & 0.863\\
    Wolf & 0.440 & 0.162 & 0.169 & \textbf{0.072} & 0.092 & 2.244 & 1.009 & 1.343 & \textbf{0.422} & 0.598\\
    Horse & 0.796 & 0.129 & 0.256 & \textbf{0.043} & 0.089 & 3.331 & 1.687 & 2.761 & \textbf{0.602} & 3.502\\
    Michael & 3.142 & 0.089 & 1.170 & \textbf{0.067} & 0.081 & 8.525 & 0.506 & 3.800 & \textbf{0.418} & 0.563\\
    Victoria & 1.302 & 0.053 & 0.512 & \textbf{0.025} & 0.041 & 7.020 & 1.030 & 2.576 & \textbf{0.404} & 2.584\\
    \hline
    \end{tabular}} 
\label{table: tosca_errors}
\end{table*}
\begin{figure*}[htbp]
\begin{center}
\includegraphics[width=0.9\textwidth]{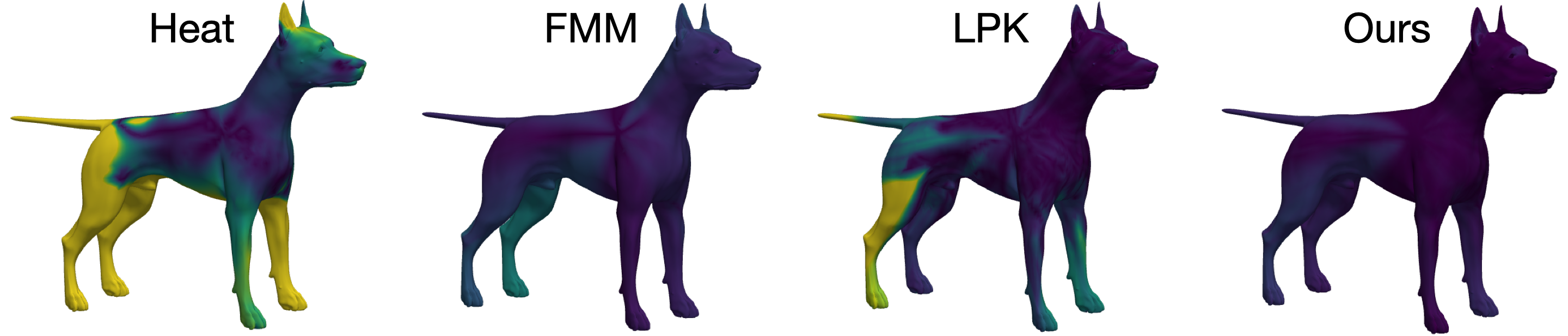}
\end{center}
  \caption{Generalization to arbitrary surfaces. 
  Errors presented for the heat kernel method, fast marching, Lichtenstein et al. and the proposed method. 
  Local errors presented as colors on the surface (brighter color indicates higher error), were computed relative to the polyhedral MMP distances.
  The evaluation was done on the TOSCA data base, whereas our solver and the solver proposed by Lichtenstein et al. were trained with a small number of $2^{nd}$ order polynomial surfaces, the three paraboloid surfaces presented in Table \ref{table: polynomial surfaces}}
\Description{Figure 8. Present errors of the Heat method, FMM, Lichtenstein et al. and our method. Fully described in the text}
\label{fig: tosca_errors}
\end{figure*}
\begin{figure*}[h!]
\begin{center}
         \centering
         \includegraphics[width=0.86\textwidth]{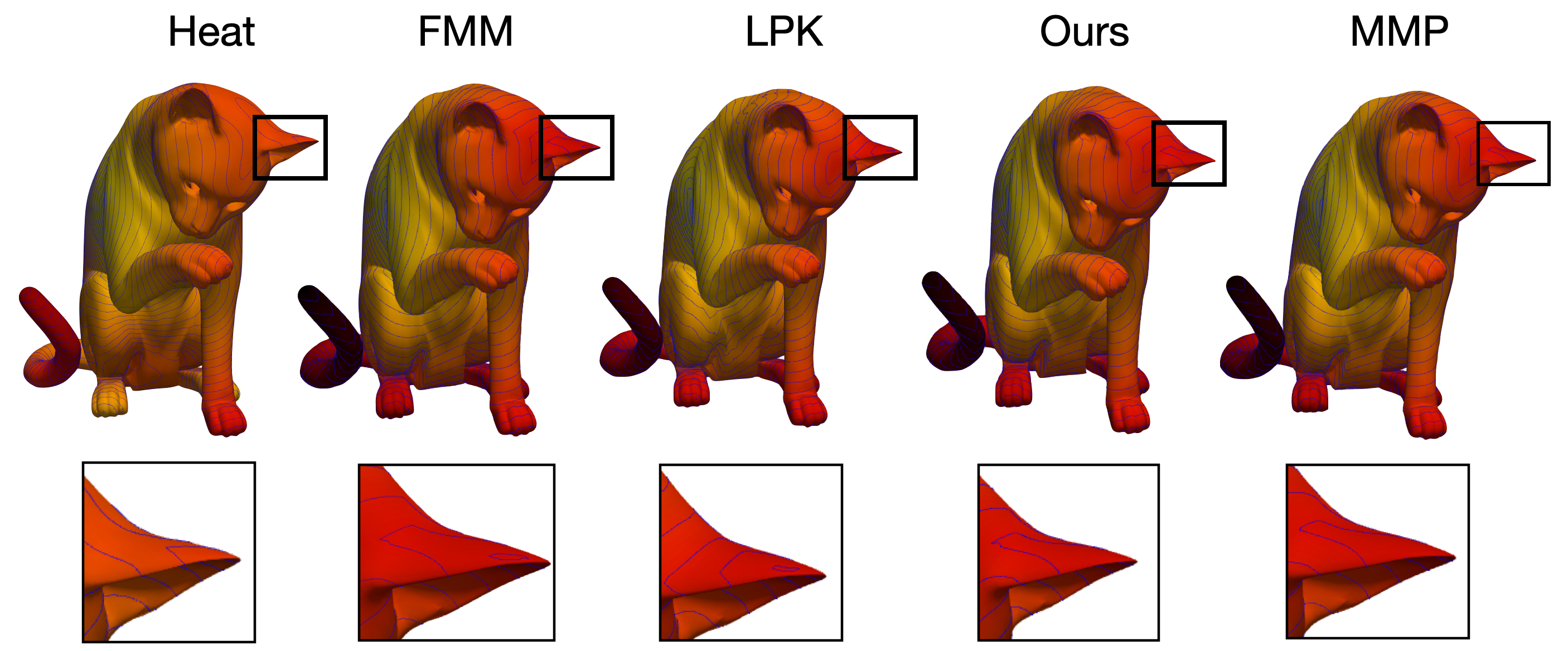}
         \includegraphics[width=0.86\textwidth]{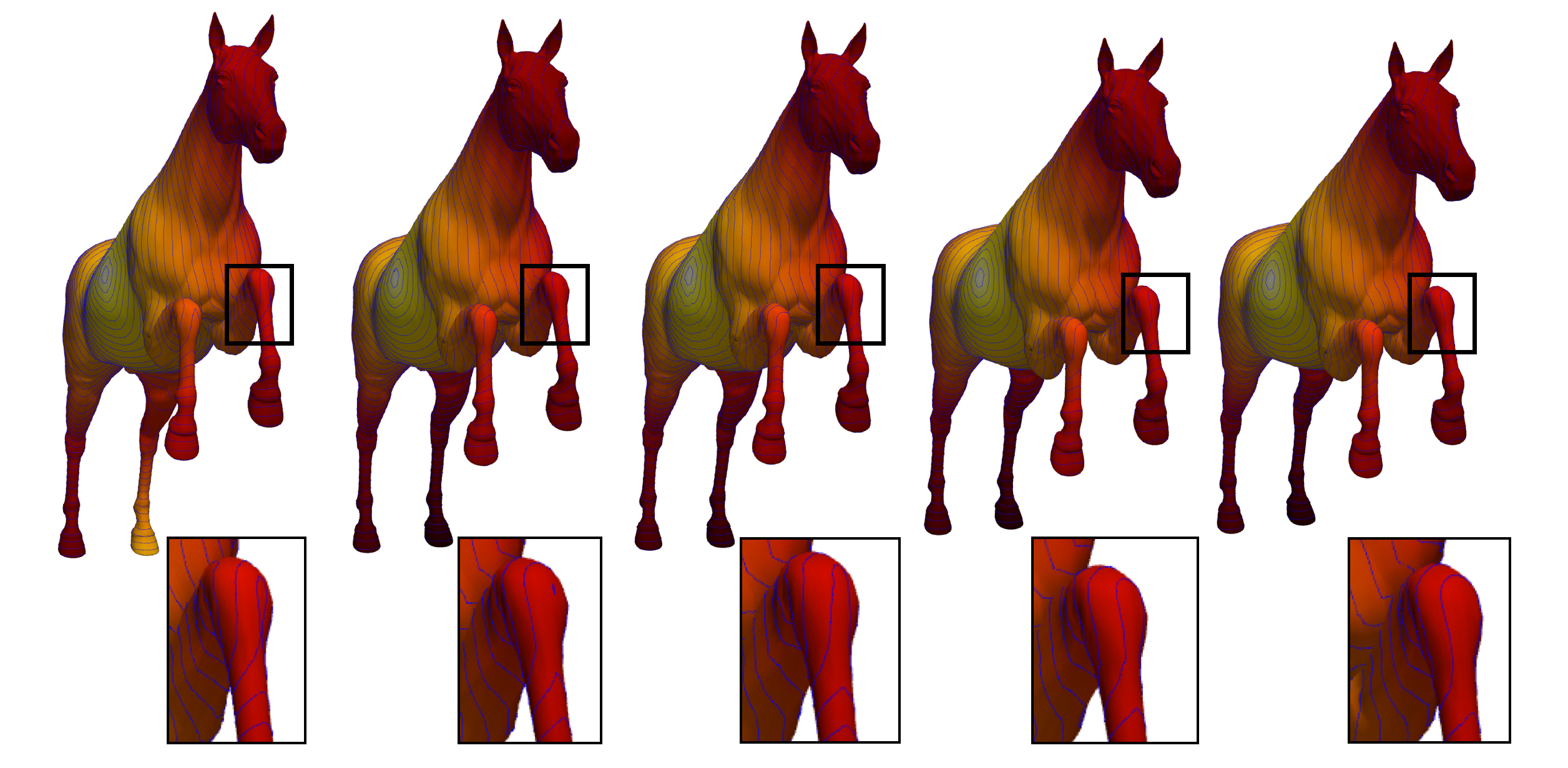}
\end{center}
\caption[Generalization to arbitrary shapes: Iso-contours shown for the heat method, fast marching, Lichtenstein et al., our method and the exact polyhedral scehme, showing the different methods performance on sharp local areas]{Generalization to arbitrary shapes: Iso-contours shown for (left to right) the heat method, fast marching, Lichtenstein et al., our method and the exact polyhedral scehme, calculated by the MMP algorithm. The evaluation was conducted on shapes from TOSCA whereas our solver and the solver proposed by Lichtenstein et al. were trained with only limited number of $2^{nd}$ order polynomial surfaces (the $3$ surfaces presented in Table \ref{table: polynomial surfaces}).
\label{fig: tosca_local_contours}
}
\Description{Figure 9. Present iso-contours on shapes from the Tosca dataset for the Heat method, FMM, Lichtenstein et al., our method and the exact polyhedral distances. Our method shows more accurate iso-contours on sharp edges, even though it was trained on smooth surfaces. Fully described in the text.}
\end{figure*}
\begin{figure*}[htbp]
\begin{center}
\includegraphics[width=1\textwidth]{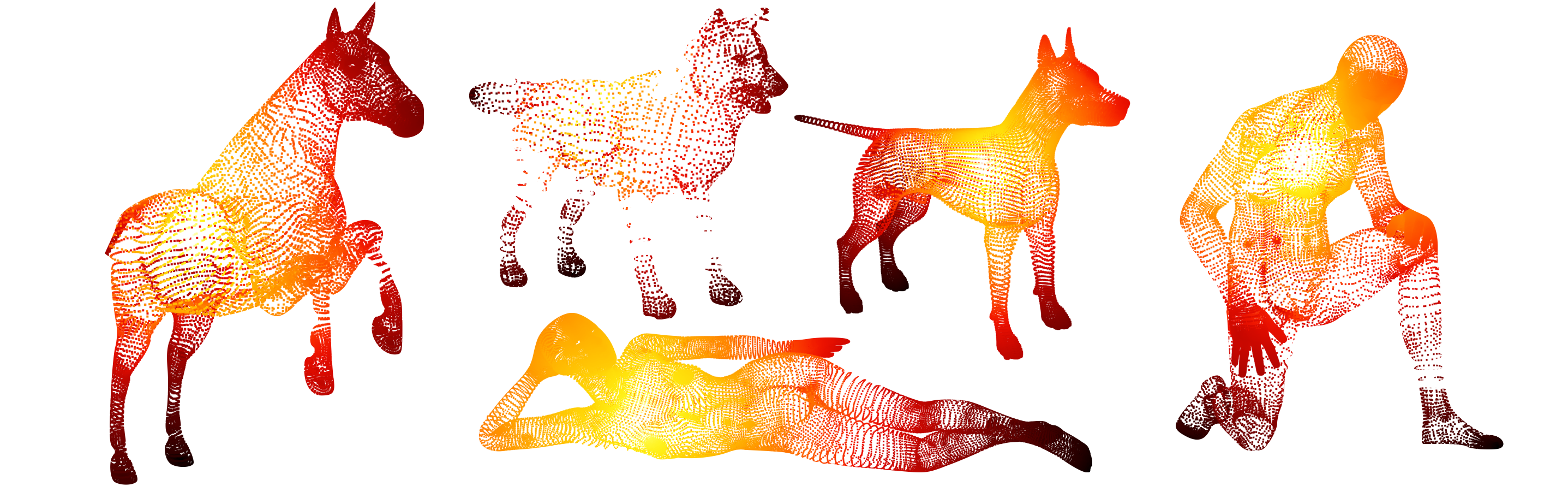}
\end{center}
  \caption[Point clouds: Distance maps computed by  the proposed method.
  The evaluation was done on shapes from the TOSCA dataset, given as point clouds.]
  {Point clouds: Distance maps computed by the proposed method.
  The evaluation was done on shapes from the TOSCA dataset, given as point clouds.
  Our solver was trained using the suggested  bootstrapping framework \ref{generate_gt_section} with a limited number of $2^{nd}$ order polynomial surfaces, represented by polygonal meshes.}
\Description{Figure 10. Point clouds. Geodesic distance computed by our method is plotted on  shapes from the Tosca dataset represented as point clouds. Fully described in the text.}
\label{fig: tosca_pointcloud}
\end{figure*}

\begin{figure}[t]
\begin{center}
    \includegraphics[width=\linewidth]{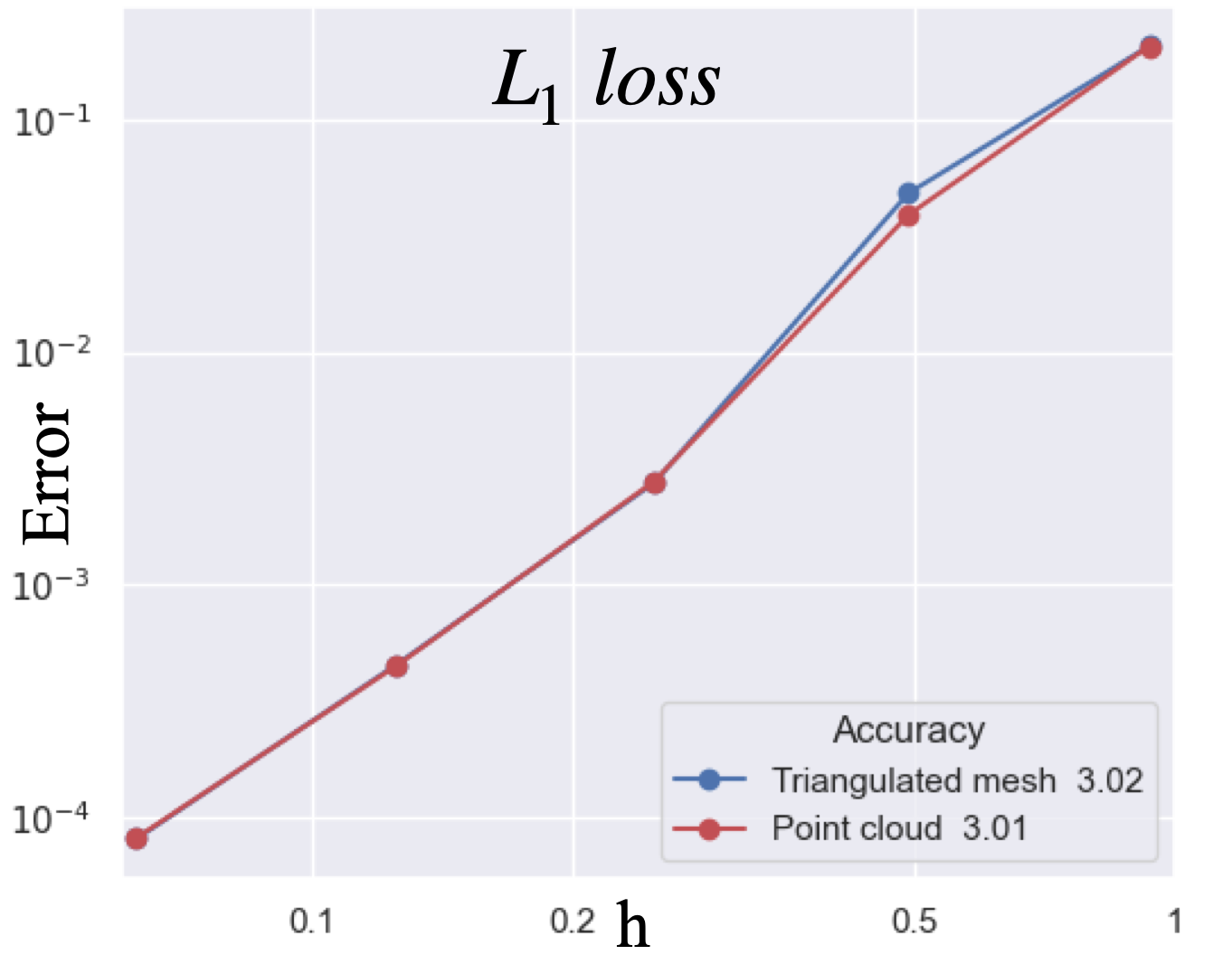}
\end{center}
   \caption{Order of accuracy for point clouds: Plots showing the effect of the edge size of the corresponding mesh or distances between nearby points in the point cloud case, on the errors. 
   When applied to point clouds, the local neighborhood is given by KNN$_{3,6}$ as defined in Subsection \ref{point_cloud_section}.}
   \Description{Figure 11. Numerical evaluation of the order of accuracy shown on spheres. The variation of our method for point clouds shows a similar order of accuracy and errors as the variation of our method for triangulated meshes.}
\label{fig: point_cloud_shperes}
\end{figure}

\begin{figure}[htbp]
\begin{center}
\includegraphics[width=0.9\linewidth]{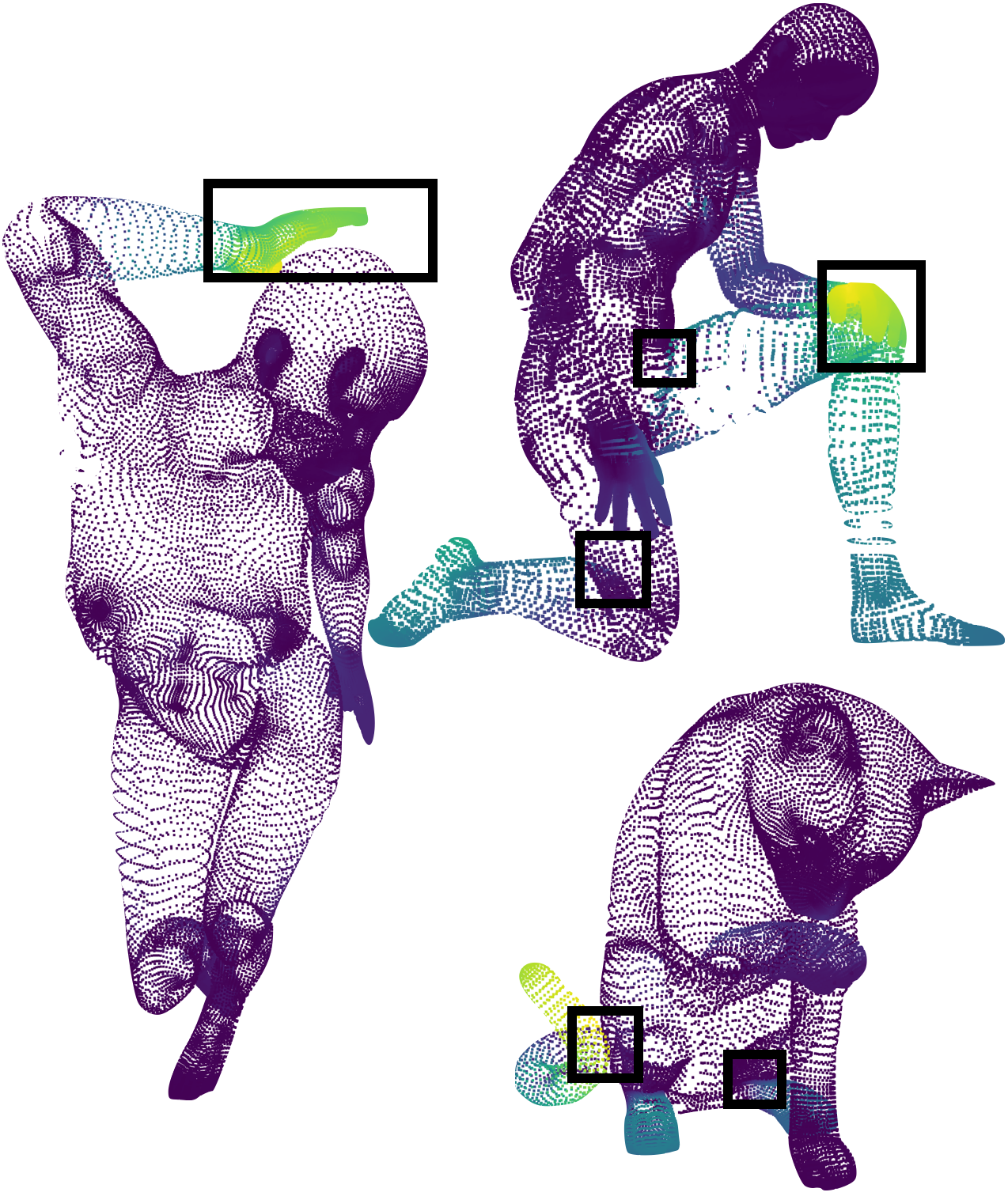}
\end{center}
  \caption[Point clouds and nearest neighbors: 
  Local inaccuracies presented as colors on the surface (brighter color indicates higher error), were computed relative to distances computed on the polyhedral mesh representation of the surfaces.]
  {Point clouds and nearest neighbors: 
  Local inaccuracies presented as colors on the surface (brighter color indicates higher error), were computed relative to distances computed on the polyhedral mesh representation of the surfaces.
    In the point cloud representation, the local neighborhood is defined using the KNN framework, which is based on the Euclidean distance between points in the extrinsic embedding space $\mathbb{R}^3$. 
  This is in contrast to the intrinsic neighborhood defined by edge connectivity in the mesh representation. 
  In the green regions, points from far away parts of the surface are wrongly identified as neighbors by the solver.
  It leads to loss of accuracy in distance approximation in such regions.
   }
   \Description{Figure 12. Point clouds and nearest neighbors. Illustration of errors that occur when defining the local neighborhood using the extrinsic embedding space $\mathbb{R}^3$. Fully described in the text.}
\label{fig: tosca_pointcloud_errors}
\end{figure}

\subsection{Generalization to polynomial surfaces}
We evaluated our method on second order polynomial surfaces.
In general, there is no closed form analytical expression for geodesic distances on such surfaces. 
To train our solver, we generated a wide variety of polynomial surfaces and obtained an accurate approximation of their geodesics for a range of sampling resolutions, as described in Section \ref{generate_gt_section}. 
After obtaining an accurate geodesic distance map, we created a training set of 100,000 examples and trained our model according to Section \ref{train_local_solver_section}.
An evaluation of our method on surfaces from this family is shown in Table \ref{table: polynomial surfaces} and Figure \ref{fig: saddle_errors}.
To confirm that our method is third order accurate on surfaces other than spheres, we performed an additional experiment in which we considered different resolutions of polynomial surfaces, see Figure \ref{fig: 3rd_order_polynom}.
In all the evaluations presented, the ground truth distances were obtained using our bootstrapping method presented in Section \ref{generate_gt_section}. 

\subsection{Generalization to arbitrary surfaces}

To better emphasize the generalization ability of our method, we conduct an additional experiment. 
We train our solver only on the three $2^{nd}$ order polynomial surfaces 
shown in Table \ref{table: polynomial surfaces}, and evaluate it on arbitrary shapes from the TOSCA dataset \cite{bronstein2008numerical}. 
It can be seen in Figure \ref{fig: tosca_errors} and Table \ref{table: tosca_errors}, that our method generalizes well and leads to significantly lower errors compared to the heat method, classical fast marching and the method presented by \cite{lichtenstein2019deep} when trained on the same polynomial surfaces. 
In  Figure \ref{fig: tosca_local_contours}, it can be seen that our method, although trained on simple and smooth surfaces, performs well even in sharp areas such as cat ears and horse knees.
Errors are computed relative to the polyhedral distances, since they are the most accurate distances available to us for these shapes. 
\begin{figure*}[htbp]
\begin{center}
    \includegraphics[width=1\textwidth]{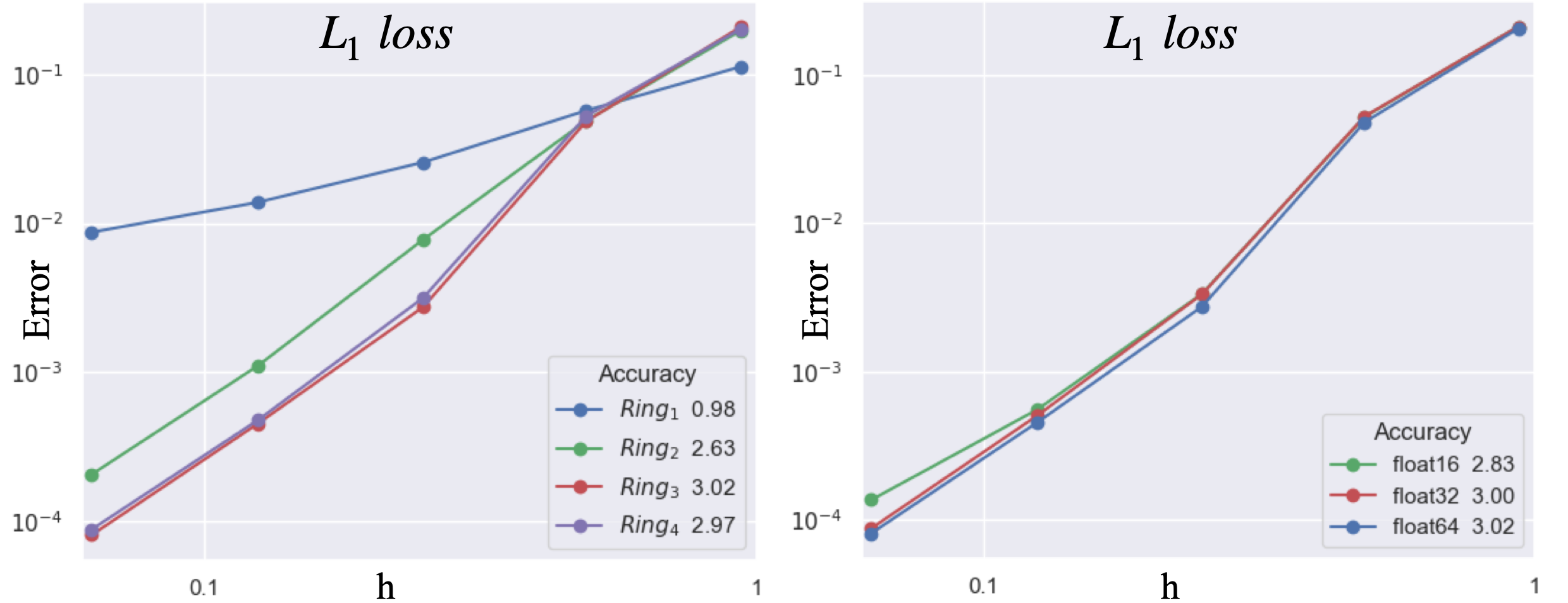}
\end{center}
\caption{Left: {Local neighborhood:} Evaluation of the proposed method on spheres with different local neighborhoods support. 
   Ring$_i$ corresponds to a neighborhood containing all vertices with at most {\it i} edges from the evaluated target.
     Right: Precision floating point representation: Evaluation of our method on spheres with different  floating point precision representation of the neural network weights.
     }
\label{fig: rings_graph}
\Description{Figure 13. Left: Local neighbourhood: numerical evaluation of the order of accuracy of our method for different local neighbourhoods. The order of accuracy for Ring1, Ring2, Ring3, and Ring4 are 0.98, 2.63, 3.02, and 2.97, respectively.
Right: Precision floating point representation: numerical evaluation of the order of accuracy of our method with different floating point representations of the neural network  weight. The order of accuracy for float16, float32, and float64 are 2.83, 3.00, and 3.02, respectively.}
\end{figure*}

\subsection{Generalization to point clouds}
\label{point_cloud_section}
The proposed local solver operates directly on surface points,
 making it well suited for working with point clouds.
To adapt our method for point clouds, we use the same distance update scheme presented in Algorithm \ref{alg:1} and define the local neighborhood with k-nearest neighbors (KNN) in the embedding space, $\mathbb{R}^3$ rather than the graph defined by the triangulation. 
Similar to our definition of Ring$_i$ neighborhood on triangulated meshes, we define local support with KNN$_{i, k}$ in a recursive manner on $i$, when $k$ is the number of neighbors in the KNN algorithm,
\begin{eqnarray}
\label{eq: KNN_recursive}
\text{KNN}_{1,k}(p) &=& \text{KNN}_k(p), \cr 
\text{KNN}_{i+1, k}(p) &=&  \text{KNN}_{i,k}(p) \bigcup \left(\bigcup_{p_j \in \text{KNN}_{i,k}(p)} \text{KNN}_{1, k}(p_j)\right)\,.
\end{eqnarray}
This change applies to two parts of our method when each part is treated separately: (1) updating the {\it Unvisited} neighbors of the last {\it Visited} point to {\it Wavefront}, and (2) defining the local support of our solver.
To update the {\it Unvisited} neighbors of the last {\it Visited} point to {\it Wavefront}, we use KNN$_{1,6}$ (regular KNN with 6-nearest neighbors).
Similarly to our 3-ring neighborhood, we  define the numerical support of our solver on point cloud by a local neighborhood of KNN$_{3,6}$.

Similar to the method proposed for triangulated meshes, we study the order of accuracy of our method for sampled spheres.
We use a local solver trained on meshes and apply it to point clouds representing sampled spheres using the suggested KNN based local support. 
Results are presented in Figure \ref{fig: point_cloud_shperes}, which shows that the proposed method is 3rd order.
In an additional experiment, we evaluate our method on surfaces from the TOSCA dataset, presented in Figure \ref{fig: tosca_pointcloud} and Table \ref{table: tosca_errors}, showing comparable results to the neighborhoods defined by triangulated meshes.

The point cloud representation does not include the connectivity between nearby points.
While the use of Euclidean distance within the KNN framework is sufficient to define  local neighborhood in most surfaces, it could fail in cases where nearest neighbors in the embedding space leads to shortcuts and topological changes of the actual shape. 
This could cause significant degradation in the approximation accuracy.
In order to reduce the effect of potential shortcuts we iterate the  neighborhood structure as described above.
This explains why, in some cases, we prefer to operate with KNN$_{2,6}$ rather than KNN$_{1,6^2+6}$, or KNN$_{3,6}$ rather than, say KNN$_{1,6^3+6^2+6}$.
\section{Ablation study}
To analyze the performance and robustness of our method, we conduct additional tests. 
These include modifying the local numerical support by which neighborhoods are defined and the precision point representation of the network weights.

\subsection{Local numerical support}.
In the fast-marching method, the solver locally estimates a solution to the eikonal equation using a finite-difference approximation of the gradient.
This approximation of the gradient is defined by a local stencil. 
For example, in the case of regularly sampled grids, the one-sided difference formula for a third order approximation requires a stencil with three neighboring points \cite{fornberg1988generation}.
In analogy to the stencil, our method uses a $3^{rd}$ ring neighborhood.
Our local solver does not explicitly solve an eikonal equation nor does it use an approximation of the gradient. 
Yet, the size of the numerical support is the underlying ingredient that allows our neural network to realize high order accuracy.
We have empirically validated this, as depicted in Figure \ref{fig: rings_graph} left.
\subsection{Numerical precision}. 
The choice of numerical representation is an important decision in the implementation of a neural network based solver.
It leads to a trade-off between the accuracy of the solver and its execution time and memory footprint.
In all our experiments, our main focus is on the accuracy of the method.
Hence, we used double precision floating point for our neural network implementation.
Figure \ref{fig: rings_graph} 
 right shows a comparison between our implemented network with different precision, showing only a slight degradation when a single and half precision floating points are used. 

\subsection{Robustness to triangulation }
The proposed local solver is applied directly to the mesh points, and does not use the underlying triangulation. 
However, the numerical support of the solver and the distance updating scheme do depend on the neighborhood graph structure induced by the triangulation.
Therefore, it is important to investigate the robustness of the proposed method to various triangulation methods.
In a less favorable scenario to the proposed method, we would like to explore the influence of ill conditioned triangles during test while training on well conditioned ones.

In the next experiment we train our neural network on regularly sampled spheres and evaluate the resulting solver on arbitrarily triangulated ones that includes ill conditioned triangles.
We show  numerical qualitative evaluation in Table \ref{table: Ill_conditioned_spheres} and Figure \ref{fig: Ill_conditioned_spheres}.
It can be seen that the proposed method is indeed robust to  different triangulations and produces lower errors compared to the heat method, the fast marching, the exact geodesic method, and the  deep-learning method of Lichtenstein et al.
\begin{table}[htbp] 
    \centering
    \captionof{table}{{Robustness to triangulation:} Quantitative evaluation of the ability of the proposed method to handle triangulations other than those for which it was trained for. 
    The evaluation was conducted on randomly uniformly non-regularly sampled spheres 
    whereas our solver and the solver proposed by Lichtenstein et al. were trained only on regularly sampled spheres.
    }
    \scalebox{1}{\begin{tabular}{|l|c|c|c|c|c|}
    \hline
    Error & FMM & Heat method & LPK & MMP & Ours\\
    \hline\hline\hline
    $L_1$ & 0.01804 & 0.01964 & 0.00510 & 0.00128 & \textbf{0.00099}\\
    $L_2$ & 0.01865 & 0.02048 & 0.00672 & 0.00142 & \textbf{0.00127}\\
    $L_\infty$ & 0.0270 & 0.0361 & 0.02278 & \textbf{0.0027} & 0.0041\\
    \hline
    \end{tabular}} 
\label{table: Ill_conditioned_spheres}
\end{table}
\begin{figure}[htbp]
\begin{center}
    \includegraphics[width=1\linewidth]{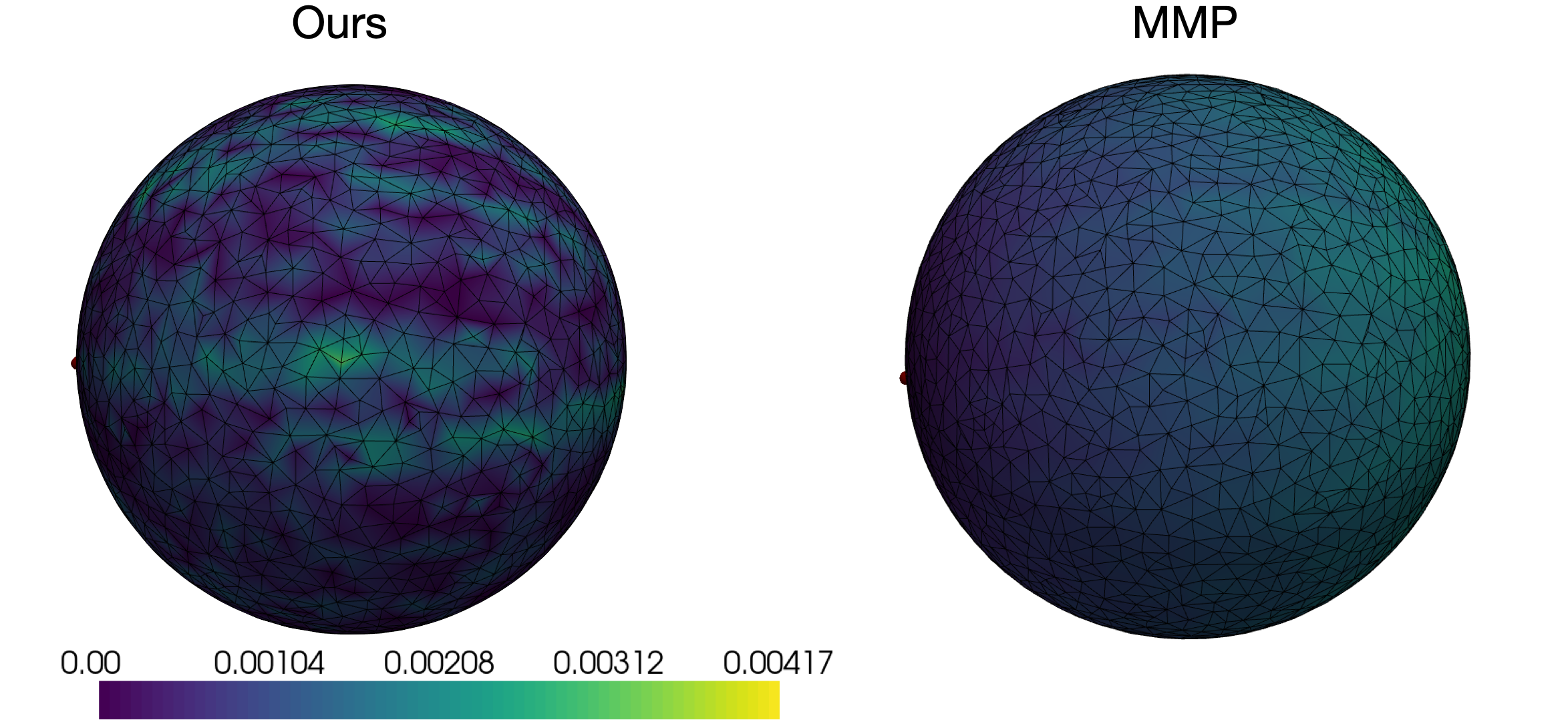}
\end{center}
   \caption[{Robustness to Triangulation:} Errors presented for the ``exact'' MMP scheme and the proposed method]{{Robustness to Triangulation:} Errors presented for the ``exact'' MMP scheme and the proposed method.  Local errors, represented as colors on the surface, were computed relative to the analytical geodesic distances. The proposed method was trained on regularly sampled spheres as can be seen in Figure \ref{fig: attention_map}.
   }
   \Description{Figure 14. Robustness to Triangulation. Visualization of errors on spheres containing ill condition triangulation. The methods presented are the MMP method and our method when trained on regularly sampled spheres. Fully described in the text.}
\label{fig: Ill_conditioned_spheres}
\end{figure}
\begin{table}[htbp] 
    \centering
    \captionof{table}{{Neural Architectures:} Various architectural modifications to the presented neural network \ref{fig:NN schematic}. 
    The residual connections and the different activation functions are defined only for the shared weight MLP, while the fully connected network applied after the pooling operation has LeakyRelu activation with a negative slope of 0.001 in all  variants of our neural network.}
    \scalebox{0.75}{\begin{tabular}{|l|c|c|c|c|c|}
    \hline
    Neural architecture & Residual connections & Activation & negative slope & Pooling\\
    \hline\hline\hline
    Model 1 & \cmark &  LeakyReLu & 0.2 & max\\
    Model 2 & \cmark &  LeakyReLu & 0.2 & avg\\
    Model 3 & \cmark & LeakyReLu & 0.001(default) & max\\
    Model 4 & \cmark & ReLu & \xmark & max\\
    Model 5 & \xmark & ReLu & \xmark & max\\
    LPK & \xmark & ReLu & \xmark & max\\
    \hline
    \end{tabular}} 
\label{table: Neural Architectures}
\end{table}
\begin{figure}[htbp]
\begin{center}
    \includegraphics[width=\linewidth]{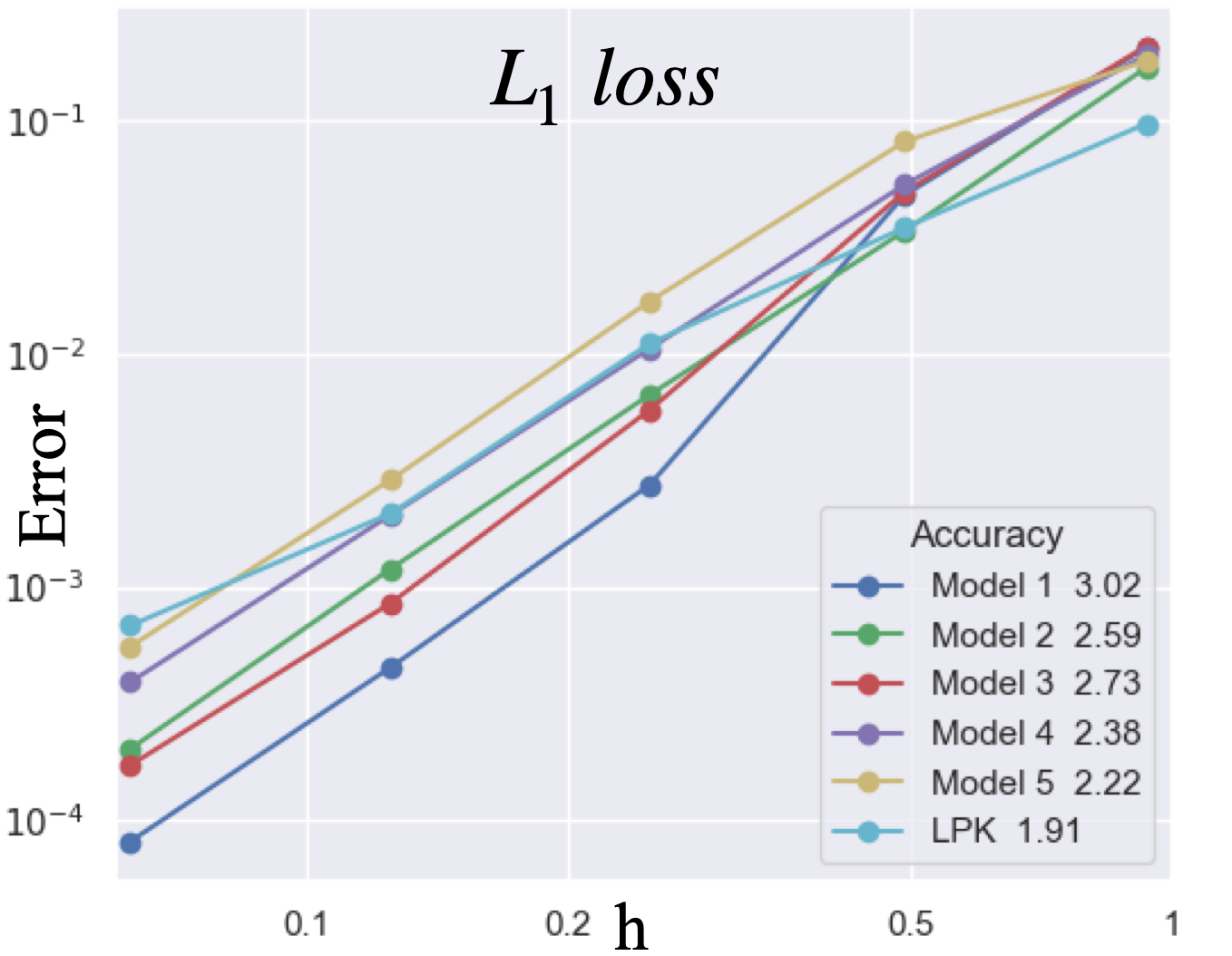}
\end{center}
   \caption[Neural architectures: Evaluation of the effects of different modifications to our neural network.]{Neural architectures: Evaluation of the effects of different modifications to our neural network, presented in \ref{table: Neural Architectures}. 
   All different models use the same local neighborhood support, which includes all vertices with at most three edges from the evaluated target.}
\label{fig: Neural Architectures}
   \Description{Figure 15. Neural architectures: numerical evaluation of the order of accuracy on spheres for different architectures presented in Table \ref{table: Neural Architectures}. The order of accuracy for Model1, Model2, Model3, Model4, and Model5 are 3.02, 2.59, 2.73, 2.38, and 2.22, respectively.}
\end{figure}

\section{Neural network architecture: Engineering considerations}

In this research our goal was to develop  high order accurate methods for computing geodesic distances on surfaces.
To overcome the 2nd order restriction of the polyhedral approximation, we start with the method presented in \cite{lichtenstein2019deep}.
This method uses a neural network based local solver that circumvents the polyhedral representation of the surface and operates directly on the neighboring points.
In the fast-marching method, the local solver estimates a solution to the eikonal equation using a finite-difference approximation of the gradient. 
The stencil size in the finite difference method determines the accuracy of the operator \cite{fornberg1988generation}. 
Our hypotheses is that extending the numerical support of the local solver would improve the overall order of accuracy.

We first considered the neural network introduced in \cite{lichtenstein2019deep} while extending the local support to $3^{rd}$ ring neighborhood.
Such a direct extension did not improve the accuracy of the method, as can be seen in Figure \ref{fig: Neural Architectures}.
Next, we studied the trained model latent vector obtained after the max-pooling operation.
Out of the $1024$ entries of this vector, only 96 were different than zero.

Based on this observation we experimented with  architectural modifications by extending the number of hidden layers, changing the activation functions, and reducing the size of the latent space. 
These modifications are presented in Table \ref{table: Neural Architectures}, where Model 1 is our proposed neural architecture presented in Figure \ref{fig:NN schematic}, and all other models have a similar architecture up to the changes detained in the table.
As can be seen in Figure \ref{fig: Neural Architectures}, the various architectural modifications have a large impact on the performance of the solver.

%


\section{Conclusions}
A fast and accurate method for computing geodesic distances on surfaces was presented. 
Revisiting the method of Lichtenstein et al. \cite{lichtenstein2019deep} we modified the ingredients of a neural network based local solver.
The solver proposed by Lichtenstein et al. was limited to second order accuracy.
The suggested improvements involved extending the local solver's numerical support and redesigning the network's architecture.
Next, we introduced a data generation mechanism that provides accurate high-order distances to train our solver by.
It allowed us to use surfaces for which there is no analytic way to compute geodesic distances.
We trained our solver using examples generated by a novel multi-resolution bootstrapping technique that projects distances computed at high resolutions to lower ones.
We believe that the proposed bootstrapping idea could be utilized for training other numerical solvers while keeping in mind that the numerical support enables the required accuracy.
We trained a neural network to locally extrapolate the values of the distance function on sampled surfaces.
The result is the most accurate method that runs at the lowest (quasi-linear) computational complexity, compared to all existing geodesic distance computation methods. 


\begin{figure}[t]
\begin{center}
    \includegraphics[width=0.7\linewidth]{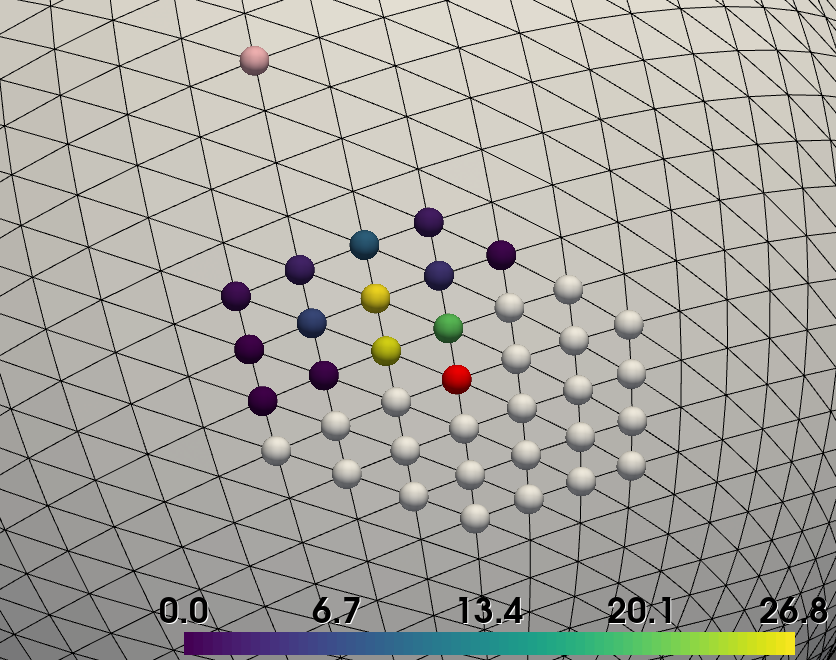}
\end{center}
   \caption[Attention map showing the importance of each neighboring point in the network approximation]{Attention map showing the importance of each neighboring point in the network approximation. 
   The source point is colored pink, the point whose distance is currently evaluated is colored red, and all other plotted points form the current local neighborhood.
   The {\it Visited} neighboring points that are fed into the neural network are colored according to the percentage of the corresponding features in the latent vector obtained from the max-pooling operation. 
   It can be seen that the neural network pays more attention to the neighboring points that are closer to the geodesic path.
   }
\label{fig: attention_map}
\Description{Figure 16. Attention map showing the importance of each neighboring point in the network approximation. Fully described in the text.}
\end{figure}


\begin{acks}
The research was partially supported by the D. Dan and Betty Kahn Michigan-Israel Partnership for Research and Education, run by the Technion Autonomous Systems and Robotics Program.
We are grateful to Mr. Alon Zvirin for all his efforts in improving the writeup and overall presentation of the ideas introduced in this paper.
\end{acks}

\bibliographystyle{ACM-Reference-Format}
\bibliography{mybibliography}


\end{document}